\documentclass[reqno,a4paper,11pt]{article}
\pdfoutput=1
\usepackage{xcolor}

\usepackage{graphicx}
\usepackage[textwidth = 430 pt, textheight = 630 pt]{geometry}

\definecolor{MyDarkBlue}{rgb}{0.15,0.25,0.45}
\usepackage{epsfig,rotating}
\usepackage{amsmath,amssymb}
\usepackage{amsfonts}
\usepackage{mathrsfs}
\usepackage{bbm}
\usepackage[normalem]{ulem}

\usepackage{latexsym}
\usepackage{amsthm}

\usepackage[utf8x]{inputenc}

\usepackage{hyperref}
\hypersetup{
hypertexnames=false,
colorlinks=true,
citecolor=MyDarkBlue,
linkcolor=MyDarkBlue,
urlcolor=MyDarkBlue,
pdfauthor={Getachew Alemu Demessie and Christian S\"amann},
pdftitle={Higher Gauge Theory with String 2-Groups},
pdfsubject={hep-th math-ph},
breaklinks=true
}

\usepackage{tikz}
\usepackage{mathtools}
\usepackage[all,knot]{xy}
\xyoption{arc}

\let\fn\footnote
\renewcommand{\footnote}[1]{\linespread{1.1}\fn{#1}\linespread{1.29}}

\makeatletter\renewcommand{\section}{\@startsection
{section}{1}{\z@}{-3.5ex plus -1ex minus
    -.2ex}{2.3ex plus .2ex}{\bf }}
\makeatletter\renewcommand{\subsection}{\@startsection{subsection}{2}{\z@}{-3.25ex
plus -1ex minus
   -.2ex}{1.5ex plus .2ex}{\bf }}
\makeatletter\renewcommand{\subsubsection}{\@startsection{subsubsection}{3}{-2.45ex}{-3.25ex
plus -1ex minus -.2ex}{1.5ex plus .2ex}{\it }}
\renewcommand{\thesection}{\arabic{section}}
\renewcommand{\thesubsection}{\arabic{section}.\arabic{subsection}}
\renewcommand{\@seccntformat}[1]{\@nameuse{the#1}.~~}

\renewcommand{\theequation}{\thesection.\arabic{equation}}
\makeatletter \@addtoreset{equation}{section}
\def\Ddots{\mathinner{\mkern1mu\raise\p@
\vbox{\kern7\p@\hbox{.}}\mkern2mu
\raise4\p@\hbox{.}\mkern2mu\raise7\p@\hbox{.}\mkern1mu}}
\usepackage[toc,page]{appendix}

\newtheorem{thm}{Theorem}[section]
\renewcommand{\thethm}{\thesection.\arabic{thm}}
\newtheorem{lemma}[thm]{Lemma}
\newtheorem{definition}[thm]{Definition}
\newtheorem{theorem}[thm]{Theorem}
\newtheorem{proposition}[thm]{Proposition}
\newtheorem{corollary}[thm]{Corollary}
\newtheorem{remark}[thm]{Remark}
\newtheorem{example}[thm]{Example}

\renewcommand{\appendices}{
\section*{Appendix}\label{appendices}\setcounter{subsection}{0}
\addcontentsline{toc}{section}{Appendix}
\setcounter{equation}{0}
\makeatletter
\renewcommand{\theequation}{\Alph{subsection}.\arabic{equation}}
\renewcommand{\thesubsection}{\Alph{subsection}}
\renewcommand{\thethm}{\Alph{subsection}.\arabic{thm}}
\@addtoreset{equation}{subsection}
\@addtoreset{thm}{subsection}
\makeatother
}

\def\slasha#1{\setbox0=\hbox{$#1$}#1\hskip-\wd0\hbox to\wd0{\hss\sl/\/\hss}}

\def\periodb#1{\setbox0=\hbox{$#1$}#1\hskip-\wd0\hbox to\wd0{-}}

\newcommand{\unit}{\mathbbm{1}}   			
\newcommand{\id}{\mathrm{id}}   			

\newcommand{\CA}{\mathcal{A}}    			

\newcommand{\CC}{\mathcal{C}}
\newcommand{\CCC}{\mathscr{C}}

\newcommand{\CD}{\mathcal{D}}
\newcommand{\CCD}{\mathscr{D}}

\newcommand{\CF}{\mathcal{F}}

\newcommand{\CG}{\mathcal{G}}

\newcommand{\CH}{\mathcal{H}}

\newcommand{\CP}{\mathcal{P}}

\newcommand{\vC}{\check{C}}

\newcommand{\CS}{\mathcal{S}}
\newcommand{\CCS}{\mathscr{S}}

\newcommand{\CU}{\mathcal{U}}

\newcommand{\CX}{\mathcal{X}}

\newcommand{\CE}{\mathcal{E}}
\newcommand{\frg}{\mathfrak{g}}				

\newcommand{\fru}{\mathfrak{u}}

\newcommand{\FR}{\mathbbm{R}}     			
\newcommand{\FC}{\mathbbm{C}}     			
\newcommand{\NN}{\mathbbm{N}}     			
\newcommand{\RZ}{\mathbbm{Z}}     			

\newcommand{\dd}{\mathrm{d}}     			
\newcommand{\dpar}{\partial}     			
\newcommand{\di}{\mathrm{i}}     			
\newcommand{\eps}{{\varepsilon}}			

\newcommand{\sB}{\mathsf{B}}

\newcommand{\eand}{{\qquad\mbox{and}\qquad}}     		
\newcommand{\ewith}{{\qquad\mbox{with}\qquad}}

\newcommand{\der}[1]{\frac{\dpar}{\dpar #1}}   		
\newcommand{\dder}[1]{\frac{\dd}{\dd #1}}   		
\newcommand{\pr}{\mathsf{pr}}     			

\newcommand{\au}{\mathfrak{u}}

\newcommand{\aspin}{\mathfrak{spin}}

\newcommand{\sU}{\mathsf{U}}     			

\newcommand{\sA}{\mathsf{A}}
\newcommand{\sG}{\mathsf{G}}

\newcommand{\sL}{\mathsf{L}}
\newcommand{\sHom}{\mathsf{Hom}}
\newcommand{\sLie}{\mathsf{Lie}}

\newcommand{\sS}{\mathsf{S}}

\newcommand{\sH}{\mathsf{H}}
\newcommand{\sSU}{\mathsf{SU}}

\newcommand{\sO}{\mathsf{O}}

\newcommand{\CatMan}{\mathsf{Mfd}^\infty}

\newcommand{\sSO}{\mathsf{SO}}
\newcommand{\sSpin}{\mathsf{Spin}}
\newcommand{\sString}{\mathsf{String}}

\newcommand{\sExt}{\mathsf{Ext}\,}

\newcommand{\acton}{\vartriangleright}     			
\def\tyng(#1){\hbox{\tiny$\yng(#1)$}}			
\def\tyoung(#1){\hbox{\tiny$\young(#1)$}}			

\newcommand{\beq}{\begin{eqnarray}}
\newcommand{\eeq}{\end{eqnarray}}

\newcommand{\djunion}{\sqcup\,}

\newcommand{\sft}{{\sf t}}

\newcommand{\sfd}{{\sf d}}

\newcommand{\sff}{{\sf f}}

\newcommand{\myxymatrix}[1]{\vcenter{\vbox{\xymatrix{#1}}}}

\newcommand{\sfa}{\mathsf{a}}

\newcommand{\sfr}{\mathsf{r}}
\newcommand{\sfl}{\mathsf{l}}
\newcommand{\sfm}{\mathsf{m}}
\newcommand{\sfid}{\mathsf{id}}
\newcommand{\sfs}{\mathsf{s}}

\newcommand{\sfe}{\mathsf{e}}

\newcommand{\inv}{\mathsf{inv}}

\newcommand{\CatSet}{\mathsf{Set}}

\newcommand{\CatManCat}{\mathsf{Mfd^\infty Cat}}

\newcommand{\CatBibun}{\mathsf{Bibun}}

\begin{document}

\begin{titlepage}
\begin{flushright}
 EMPG--16--07
\end{flushright}
\vskip 2.0cm
\begin{center}
{\LARGE \bf Higher Gauge Theory with String 2-Groups}
\vskip 1.5cm
{\Large Getachew Alemu Demessie and Christian S\"amann}
\setcounter{footnote}{0}
\renewcommand{\thefootnote}{\arabic{thefootnote}}
\vskip 1cm
{\em Maxwell Institute for Mathematical Sciences\\
Department of Mathematics, Heriot--Watt University\\
Colin Maclaurin Building, Riccarton, Edinburgh EH14 4AS, U.K.}\\[0.5cm]
{Email: {\ttfamily gd132@hw.ac.uk~,~c.saemann@hw.ac.uk}}
\end{center}
\vskip 1.0cm
\begin{center}
{\bf Abstract}
\end{center}
\begin{quote}
  We give a complete and explicit description of the kinematical data of higher gauge theory on principal 2-bundles with the string 2-group model of Schom\-mer-Pries as structure 2-group. We start with a self-contained review of the weak 2-category $\CatBibun$ of Lie groupoids, bibundles and bibundle morphisms. We then construct categories internal to $\CatBibun$, which allow us to define principal 2-bundles with 2-groups internal to $\CatBibun$ as structure 2-groups. Using these, we Lie-differentiate the 2-group model of the string group and we obtain the well-known string Lie 2-algebra. Generalizing the differentiation process, we find Maurer--Cartan forms leading us to higher non-abelian Deligne cohomology, encoding the kinematical data of higher gauge theory together with their (finite) gauge symmetries. We end by discussing an example of non-abelian self-dual strings in this setting.
\end{quote}
\end{titlepage}

\tableofcontents

\section{Introduction and results}

Higher gauge theory~\cite{Baez:2002jn,Baez:2004in} is an extension of gauge theory which allows for a consistent and non-abelian parallel transport of extended objects, avoiding various na\"ive no-go theorems~\cite{Baez:2010ya}. It is particularly interesting in the context of string theory, as it may be a good starting point for developing a description of M5-branes, see e.g.~\cite{Fiorenza:2012tb}. 

One of the most important open problems in this context is the lack of solutions to higher gauge equations which are truly non-abelian. More specifically, no higher principal bundle with connection is known that is not gauge equivalent to a trivially embedded abelian gerbe with connection. This is particularly unfortunate because knowing such a solution would lead to immediate progress in higher gauge theory, both on the mathematical and the physical side.

Obvious solutions to look for are higher gauge theoretic versions of monopoles and instantons. Indeed, higher twistor descriptions of potential such solutions have been successfully developed~\cite{Saemann:2012uq,Saemann:2013pca,Jurco:2014mva}, but these have not led to new solutions so far. Candidates for non-abelian self-dual string solutions within higher gauge theory were constructed in~\cite{Palmer:2013haa}, but these have the disadvantage that they either do not satisfy the so-called fake-curvature condition\footnote{This condition guarantees that the parallel transport of extended objects is invariant under reparameterizations. One might, however, argue that for the simplest self-dual strings, this condition becomes irrelevant.} or partially break the original gauge symmetry of the higher principal bundle.

It is therefore important to consider generalizations of the current formulations of higher gauge theory which do allow for interesting solutions. One such generalization has been proposed in~\cite{Ritter:2015zur}, where spacetime was replaced by a categorified space. Here, we develop higher gauge theory with smooth 2-groups, which are 2-groups internal to the weak 2-category $\CatBibun$ of Lie groupoids, bibundles and bibundle morphisms. We focus our attention in particular on the smooth 2-group model of the string group given by Schommer--Pries~\cite{Schommer-Pries:0911.2483}. 

This 2-group model of the string group is interesting for a number of reasons. First, recall that the most relevant examples of non-abelian monopoles on $\FR^3$ and instantons on $\FR^4$ form connections on principal bundles with structure group $\sSU(2)$, where this gauge group is intrinsically linked to the spin groups $\sSpin(3)\cong \sSU(2)$ and $\sSpin(4)\cong \sSU(2)\times \sSU(2)$ of the isotropy groups $\sSO(3)$ and $\sSO(4)$ of the underlying spacetimes. Correspondingly, one might expect the higher version of the spin group, the string group, to be relevant in the description of higher monopoles and instantons. Other evidence originating from an analysis of the topological part of the M5-brane world-volume action~\cite{Fiorenza:2012tb} suggest that the string group of $E_8$ might be the appropriate choice. There is further ample motivation from both physics and mathematics for being interested in higher gauge theory with the string group, stemming from the connection to 2-dimensional supersymmetric sigma-models and elliptic cohomology. For a more detailed account, see~\cite{Schommer-Pries:0911.2483}.

Our goal in this paper is thus to describe explicitly the kinematical data of higher gauge theory with the smooth string 2-group model as gauge symmetry structure. In particular, we will need to develop the appropriate notion of principal bundle, connection and corresponding gauge transformations. We intend to use our results as a starting point for finding higher monopole and instanton solutions in future work.

Principal bundles with smooth structure 2-group can be defined in (at least) two ways. First, we can regard them as certain smooth stacks over the base manifold, as done in~\cite{Schommer-Pries:0911.2483} and we review and explain this definition in our paper. Second, we can give a description in terms of generalized cocycles with values in the string 2-group. This requires us to introduce the notion of a category internal to the weak 2-category $\CatBibun$ together with weak internal functors. The resulting internal category trivially contains ordinary categories internal to the category $\CatMan$ of smooth manifolds. We can then define a principal 2-bundle as a weak functor from the \v Cech groupoid of the relevant cover of the base manifold to the delooping of the smooth 2-group. Both approaches are equivalent, but we will mostly use the latter one as it leads to a convenient description of gauge theory.

As shown in~\cite{Zhu:1204.5583}, the notion of smooth 2-group is in fact equivalent to a Lie quasigroupoid, or $(2,0)$-category internal to $\CatMan$, with a single object, which is defined in terms of Kan simplicial manifolds. As far as we are aware, this is the most general reasonable notion of Lie 2-group available in the literature today. In particular, smooth 2-groups contain ordinary groups as well as strict 2-groups. Our notion of principal 2-bundle is therefore very comprehensive and as we show in some detail, special cases include ordinary principal bundles as well as principal bundles with strict structure 2-groups. 

Introducing a connection on these principal 2-bundles is more work, as it also involves the Lie 2-algebra of the underlying smooth 2-group. To simplify our computations, we lift the string 2-group model to a weak 2-group model by introducing preferred horn fillers in the underlying Kan simplicial manifold.

Having defined principal smooth 2-group-bundles, we can readily use an approach by \v Severa~\cite{Severa:2006aa} for this purpose. Here, the higher Lie algebra of a higher Lie group arises from the moduli space of functors from the category of manifolds $\CatMan$ to descent data for principal bundles with the higher Lie group as structure group over surjective submersions $N\times \FR^{0|1}\twoheadrightarrow N$, $N\in \CatMan$. Following this approach, we successfully differentiate the smooth 2-group model for the string group and the resulting Lie 2-algebra is indeed the well-known string Lie 2-algebra.

Given a Lie 2-algebra, we can immediately derive the local description of higher gauge theory with the string 2-group together with infinitesimal gauge transformations from appropriate homotopy Maurer--Cartan equations on some $L_\infty$-algebra and their infinitesimal symmetries. 

To glue together these local connection forms to global objects, however, we also need the explicit form of finite gauge transformations. These can be obtained by extending \v Severa's differentiation approach. Coboundaries between the descent data for equivalent principal bundles induce equivalence relations on the moduli space of functors, which directly translate into finite gauge transformations of the connection forms. From these we can glean a full description of the kinematical data of higher gauge theory with the string 2-group. Put in mathematical terms, we obtained a very explicit description of the second Deligne cohomology group with values in the smooth string 2-group model.

As an application, we discuss examples of solutions to the non-abelian self-dual string equations. Due to the form of the string Lie 2-algebra, these solutions still reduce to the well-known abelian ones, if the fake curvature condition is imposed. A more comprehensive study of self-dual string solutions using smooth 2-groups is postponed to future work.

In this paper, we have tried to be rather self-contained in our presentation to facilitate access to concepts and methods that might not be very well-known as of now, such as the weak 2-category $\CatBibun$, Segal--Mitchison group cohomology and the extension of \v Severa's differentiation process leading to gauge potentials and their finite gauge theories.

Finally, a remark on our notation. As in~\cite{Schommer-Pries:0911.2483}, we work with left-principal bibundles which encode morphisms from a Lie groupoid $\CH$ to some Lie groupoid $\CG$ and for which the (left-) action of the morphisms of $\CG$ onto the bibundles is principal. In general, we try to use a consistent right-to-left notation, but we still write $B:\CH\rightarrow \CG$ for a bibundle from $\CH$ to $\CG$ as well as $B\in \CatBibun(\CH,\CG)$.

\section{The weak 2-category \texorpdfstring{$\CatBibun$}{Bibun}} \label{Section_Bibun}

\subsection{Bibundles as morphisms between Lie groupoids}\label{ssec:smooth-2-groups}

To discuss gauge theories, we will have to describe various group actions on fields. Such actions are most naturally captured in the language of groupoids. 
\begin{definition}
 A \uline{groupoid} is a small category, in which every morphism is invertible.
\end{definition}
\noindent The idea here is that the objects of a groupoid describe a set that the morphisms act on. A prominent example is the groupoid arising from the action of a group $\sG$ on a set $X$. The action groupoid $X//\sG$ has objects $X$ and morphisms $X\times \sG$. We define source and target maps on the morphisms as $\sfs(x,g)=x$ and $\sft(x,g)=g\acton x$; the identities are given by $\id_x=(x,\unit_\sG)$. Composition is defined for pairs $(x,g)$ and $(\tilde x,\tilde g)$, if $\sft(x,g)=\tilde x$ and we then have $(x,g)\circ (\tilde x,\tilde g)=(x,g\tilde g)$. 

If we are merely interested in the group $\sG$ itself, we can consider the case where $X$ is the one-element set $X=*$ on which $\sG$ acts trivially. This yields the so-called delooping $\sB\sG=(\sG\rightrightarrows*)$ of a group $\sG$, with the elements of $\sG$ forming the morphisms. Composition of morphisms is here the group multiplication and the embedding of $*$ in the morphisms yields the unit in~$\sG$.

To define groupoids with more structure, we use the concept of {\em internalization}. Essentially, the objects and morphisms of a category internal to a category $\CCC$ are objects of $\CCC$, while the structure maps consisting of source, target and composition are morphisms of $\CCC$. In particular, we can consider groupoids internal to $\CatMan$, the category of smooth manifolds and smooth morphisms between them.
\begin{definition}
 A \uline{Lie groupoid} is a groupoid internal to $\CatMan$. 
\end{definition}
\noindent That is, the objects $\CG_0$ and morphisms $\CG_1$ of a Lie groupoid $\CG$ are smooth manifolds and the structure maps $\sfs,\sft,\circ,\id$ are all smooth. Since $\CatMan$ does not have all pullbacks, we also have to demand that $\sfs$ and $\sft$ are (surjective) submersions. Otherwise, the domain of the composition morphism, $\CG_1\times^{\sft,\sfs}_{\CG_0}\CG_1$, might not be a manifold. A ubiquitous example of a Lie groupoid is the delooping $\sB\sG=(\sG\rightrightarrows *)$ of a Lie group $\sG$. 

Lie groupoids and the functors internal to $\CatMan$ between these form the category of Lie groupoids. We will be interested in an extension of this category to a weak 2-category, in which the morphisms between groupoids are generalized to bibundles.
\begin{definition}
 A \uline{(left-) principal bibundle} from a Lie groupoid $\CH=(\CH_1\rightrightarrows \CH_0)$ to a Lie groupoid $\CG=(\CG_1\rightrightarrows \CG_0)$ is a smooth manifold $B$ together with a smooth map $\tau:B\rightarrow \CG_0$ and a surjective submersion $\sigma:B\twoheadrightarrow\CH_0$ 
\begin{equation}
  \xymatrixcolsep{2pc}
  \xymatrixrowsep{2pc}
  \myxymatrix{ \CG_1\ar@<-.5ex>[dr] \ar@<.5ex>[dr] & & \ar@{->}[dl]_{\tau} B\ar@{->>}[dr]^{\sigma} & & \ar@<-.5ex>[dl] \ar@<.5ex>[dl]\CH_1\\
  & \CG_0 &  & \CH_0
    }
\end{equation}
 Moreover, there are left- and right-action maps 
 \begin{equation}
  \CG_1{}\times^{\sfs,\tau}_{\CG_0} B\rightarrow B\eand B{}\times^{\sigma,\sft}_{\CH_0} \CH_1\rightarrow B~,
 \end{equation}
 which satisfy the following compatibility relations
 \begin{itemize}
  \item[(i)] $g_1(g_2b)=(g_1g_2)b$ for all $(g_1,g_2,b)\in \CG_1\times^{\sfs,\sft}_{\CG_0}\CG_1\times^{\sfs,\tau}_{\CG_0}B$;
  \item[(ii)] $(bh_1)h_2=b(h_1h_2)$ for all $(b,h_1,h_2)\in B\times^{\sigma,\sft}_{\CH_0}\CH_1\times^{\sfs,\sft}_{\CH_0}\CH_1$;
  \item[(iii)] $b\;\id_\CH(\sigma(b))=b$ and $\id_\CG(\tau(b))\;b=b$ for all $b\in B$;
  \item[(iv)] $g(bh)=(gb)h$ for all $(g,b,h)\in \CG_1\times^{\sfs,\tau}_{\CG_0} B\times^{\sigma,\sft}_{\CH_0} \CH_1$;
  \item[(v)] The map $\CG_1\times^{\sfs,\tau}_{\CG_0}B\rightarrow B\times_{\CH_0}B~:~(g,b)\mapsto (gb,b)$ is an isomorphism (and thus the $\CG_1$-action is transitive).
 \end{itemize}
\end{definition}
\noindent Analogously, one defines right-principal bibundles. All bibundles in this paper will be left-principal bibundles and we will always clearly mark the surjections $\sigma$ in our diagrams by a two-headed arrow $\twoheadrightarrow$.

The generalized maps between Lie groupoids encoded in (equivalence classes of) bibundles are also called {\em Hilsum--Skandalis morphisms}~\cite{MR925720}. The maps $\sigma$ and $\tau$ should be regarded as source and target maps and the morphisms between $h_0\in \CH_0$ and $g_0\in \CG_0$ are given by elements $b\in B$ with $\sigma(b)=h_0$ and $\tau(b)=g_0$. The morphisms between morphisms $h_1\in \CH_1$ and $g_1\in \CG_1$ are then given by the principal left-action. We will return to this point shortly.

Bibundles contain ordinary functors between Lie groupoids as follows.
\begin{definition}
 Consider a morphism of Lie groupoids $\phi=(\phi_0,\phi_1)$ between Lie groupoids $\CH$ and $\CG$, $\phi_{0,1}:\CH_{0,1}\rightarrow \CG_{0,1}$. The \uline{bundlization} $\hat\phi$ of $\phi$ is the bibundle
\begin{equation}
  \xymatrixcolsep{2pc}
  \xymatrixrowsep{2pc}
  \myxymatrix{ \CG_1\ar@<-.5ex>[dr] \ar@<.5ex>[dr] & & \ar@{->}[dl]_{\sft} \CH_0\times^{\phi_0,\sfs}_{\CG_0} \CG_1\ar@{->>}[dr]^{\pi} & & \ar@<-.5ex>[dl] \ar@<.5ex>[dl]\CH_1\\
  & \CG_0 &  & \CH_0
    }
\end{equation}
where $\sft$ is the target map in $\CG$ and $\pi$ is the obvious projection. The actions of $\CG_1$ and $\CH_1$ on $\hat\phi$ are given by
\begin{equation}
 g'(x,g):=(x,g'\circ g)\eand (x,g)h:=(\sfs(h),g\circ \phi_1(h))
\end{equation}
for $g,g'\in \CG_1$, $h\in \CH_1$ and $(x,g)\in \hat\phi$.
\end{definition}
There is now a nice characterization of bibundles arising from bundlization, cf.\ e.g.~\cite{Mrcun:1996aa,Lerman:0806.4160}:
\begin{proposition}\label{prop:bundlization}
 Given a bibundle $B:\CH\rightarrow \CG$ between Lie groupoids $\CG$ and $\CH$, the map $\sigma:B\twoheadrightarrow \CH_0$ admits a smooth (global) section if and only if $B$ is isomorphic to a bundlization.
\end{proposition}
\begin{proof}
 Assume that $\hat \phi$ is a bundlization of a functor $\phi$. Then a section $\gamma:\CH_0\rightarrow B$ of $\sigma:B\twoheadrightarrow\CH_0$ is given by $\gamma(h_0)=(h_0,\id_{\phi_0(h_0)})$. Conversely, given a section $\gamma$, we define a functor $\phi=(\phi_0,\phi_1)$ by putting $\phi_0(h_0)=\tau(\gamma(h_0))$. The map on morphisms $\phi_1:\CH_1\rightarrow \CG_1$ is defined via its left-action
 \begin{equation}\label{eq:def-phi1}
  \phi_1(h_1)\gamma(\sfs(h_1)):=\gamma(\sft(h_1))h_1
 \end{equation}
 for all $h_1\in \CH_1$. Because this action is principal, this fixes $\phi(h_1)$ uniquely. Note that $\sfs(\phi_1(h_1))=\phi_0(\sfs(h_1))$ because the left-action is a map $\CG_1{}\times^{\sfs,\tau}_{\CG_0} B\rightarrow B$. Similarly, we have $\phi_1(\id_{h_0})=\id_{\phi_0(h_0)}$ due to axiom (iii) in the definition of bibundles and $\sft(\phi_1(h_1))=\sft(\phi_1(h_1)\gamma(\sfs(h_1)))=\tau(\gamma(\sft(h_1))h_1)=\tau(\gamma(\sft(h_1)))=\phi_0(\sft(h_1))$. Composition is by definition~\eqref{eq:def-phi1} compatible with the resulting functor $\phi$.
\end{proof}
\noindent Note that in the proof above, the construction of a section from a functor and that of a functor from a section are inverses of each other. In particular, if one starts from a section $\gamma$ from a functor $\phi$, the reconstruction of a functor from the bundlization $\hat\phi$ and the section $\gamma$ returns the original functor $\phi$.

\

Let us now list a few instructive examples of bibundles. We evidently have the identity bibundle from a Lie groupoid $\CH$ to itself,
\begin{equation}
  \xymatrixcolsep{2pc}
  \xymatrixrowsep{2pc}
  \myxymatrix{ \CH_1\ar@<-.5ex>[dr] \ar@<.5ex>[dr] & & \ar@{->}[dl]_{\sft} \CH_1\ar@{->>}[dr]^{\sfs} & & \ar@<-.5ex>[dl] \ar@<.5ex>[dl]\CH_1\\
  & \CH_0 &  & \CH_0
    }
\end{equation}
which is the bundlization of the identity functor of Lie groupoids. Moreover, bibundles include smooth maps between manifolds and Lie group homomorphisms via bundlization. Inversely, a bibundle between discrete\footnote{By discrete, we shall always mean categorically discrete, i.e.\ no morphisms beyond the identities, and not topologically discrete.} Lie groupoids $X\rightrightarrows X$ and $Y\rightrightarrows Y$ reduces to a morphism $X\rightarrow Y$, as condition $(v)$ in the definition implies that the total space of the bibundle is $X$. That is, bibundles between discrete Lie groupoids arise from a bundlization of smooth maps between manifolds. Similarly, bibundles between Lie groupoids $\CH=(\sH\rightrightarrows*)$ and $\CG=(\sG\rightrightarrows *)$ for Lie groups $\sH$ and $\sG$ arise from a bundlization of a smooth functor corresponding to a group homomorphism.  

Another non-trivial example is a principal $\sG$-bundle over a manifold $X$ where $\sG$ is an ordinary Lie group, which can be regarded as a bibundle between the Lie groupoids $X\rightrightarrows X$ and $\sG\rightrightarrows *$. For a very detailed review on Lie groupoid bibundles, see also~\cite{Blohmann:2007ez}.

\begin{definition}
 A \uline{bibundle map} between bibundles $B$ and $B'$ between Lie groupoids $\CH$ and $\CG$ with structure maps $(\sigma,\tau)$ and $(\sigma',\tau')$ is a map $\phi:B\rightarrow B'$, which is biequivariant. That is, $\sigma'\circ \phi=\sigma$, $\tau'\circ \phi=\tau$, and $\phi$ commutes with the $\CH$ and $\CG$ actions.
\end{definition}

Bibundles between Lie groupoids $\CH$ and $\CG$ together with bibundle maps form the category $\CatBibun(\CH,\CG)$. Note that we can also compose bibundles using the notion of coequalizer\footnote{see appendix~\ref{app:B}}. Given two bibundles $B:\CH\rightarrow \CG$ and $B':\CE\rightarrow \CH$, we have the coequalizer
\begin{equation}
 B\times^{\sigma,\sft}_{\CH_0}\CH_1\times^{\sfs,\tau}_{\CH_0}B'\rightrightarrows B\times^{\sigma,\tau}_{\CH_0}B'\rightarrow B\otimes B'~,
\end{equation}
where the maps denoted by the double arrow are the left- and right-actions of $\CH_1$ on $B'$ and $B$, respectively. The coequalizer is therefore the bibundle given by the quotient by the diagonal action, 
\begin{equation}
B\otimes B'=(B\times_{\CH_0}B')/\CH_1~,
\end{equation}
where $\CH_1$ acts on $B\times_{\CH_0}B'$ as $h:(b,b')\mapsto(bh,h^{-1}b')$. This composition is associative only up to a natural isomorphism of bibundles. For more details, see~\cite{Mrcun::235-253,Blohmann:2007ez}. 

With this composition, the categories $\CatBibun(\CH,\CG)$ extend to a weak 2-category.
\begin{proposition}
 There is a weak 2-category consisting of Lie groupoids as objects, bibundles as morphisms and bibundle maps as 2-morphisms. We denote this weak 2-category by $\CatBibun$.
\end{proposition}
Note that the strict 2-category consisting of Lie groupoids, smooth functors and natural transformations, which is a subcategory of $\CatManCat$, the strict 2-category of categories, functors and natural transformations internal to $\CatMan$, is also a subcategory of $\CatBibun$. The embedding of the objects is trivial and that of smooth functors is given by bundlization. Maps $\hat\alpha$ between bibundles $\hat \phi$ and $\hat \psi$ between Lie groupoids $\CH$ and $\CG$, 
\begin{equation}
  \xymatrixcolsep{2pc}
  \xymatrixrowsep{2pc}
  \myxymatrix{ 
  & & \hat\phi=\CH_0\times^{\phi_0,\sfs}_{\CG_0}\CG_1\ar@{=>}[dd]^{\hat\alpha} \ar@{->}[dl]_{\sft} \ar@{->>}[dr]^{\pi}& & \\
  \CG_1\ar@<-.5ex>[r] \ar@<.5ex>[r] & \CG_0 &  & \CH_0 & \CH_1\ar@<-.5ex>[l] \ar@<.5ex>[l]  \\
  & & \hat\psi=\CH_0\times^{\psi_0,\sfs}_{\CG_0}\CG_1 \ar@{->}[ul]_{\sft} \ar@{->>}[ur]^{\pi}& &
    }
\end{equation}
 are compatible with the right actions involving $\phi_1$ and $\psi_1$. Therefore, they have to be of the form $\hat\alpha:(h,g)\mapsto(h',g'):=(h,g \alpha(h))$, where $\alpha(h)\in \CG_1$ encodes a natural transformation $\alpha$ between $\phi$ and $\psi$. This directly implies the following, cf.\ e.g.~\cite{Li:2014}:
\begin{proposition}\label{prop:natural_trafos_1_1}
 Bibundle morphisms between bundlizations $\hat\phi$ and $\hat\psi$ are in one-to-one correspondence with natural transformations $\phi\Rightarrow \psi$.
\end{proposition}

We now come to the definition of equivalent Lie groupoids via weak 1-isomorphisms in $\CatBibun$.
\begin{definition}
 A \uline{bibundle equivalence} is a bibundle $B\in \CatBibun(\CH,\CG)$, which also defines a bibundle $B^{-1}\in \CatBibun(\CG,\CH)$ by reversing the roles of $\sigma$ and $\tau$. Two Lie groupoids $\CG$ and $\CH$ are \uline{equivalent}, if there is a bibundle equivalence between them.
\end{definition}
Note that the weak 2-category $\CatBibun$ can be regarded as the 2-category of ``stacky manifolds.'' In particular, Lie groupoids are presentations of smooth stacks. In this context, bibundle equivalence amounts to Morita equivalence.

As an example, consider the action groupoid $\CG=\sG\ltimes \sG\rightrightarrows \sG$ for some Lie group $\sG$. We shall see soon that this groupoid can be regarded as a Lie 2-group and corresponds to the crossed module $\sG\xrightarrow{~\sft~}\sG$. This action groupoid is Morita equivalent to the trivial Lie groupoid $*\rightrightarrows *$ and the bibundle equivalence reads as
\begin{equation}
  \xymatrixcolsep{2pc}
  \xymatrixrowsep{2pc}
  \myxymatrix{ \sG\times \sG \ar@<-.5ex>[dr] \ar@<.5ex>[dr] & & \ar@{->>}[dl]_{\tau} \sG\ar@{->>}[dr]^{\sigma} & & \ar@<-.5ex>[dl] \ar@<.5ex>[dl] {*}\\
  & \sG &  & {*}
    }
\end{equation}
where $\tau$ is the identity and $\sigma$ is trivial. For another example, consider the \v Cech groupoid $U^{[2]}\rightrightarrows U$ of a cover $U:=\sqcup_i U_i$ of a manifold $X$, where $U^{[2]}:=\sqcup_{ij} U_i\cap U_j$. This Lie groupoid is equivalent to the manifold $X$ itself:
\begin{equation}
  \xymatrixcolsep{2pc}
  \xymatrixrowsep{2pc}
  \myxymatrix{ U^{[2]}\ar@<-.5ex>[dr] \ar@<.5ex>[dr] & & \ar@{->>}[dl] U\ar@{->>}[dr]& & \ar@<-.5ex>[dl] \ar@<.5ex>[dl] X\\
  & U &  & X
    }
\end{equation}
 
\subsection{Smooth 2-groups}

The smooth 2-groups we are interested in are in fact 2-groups internal to the weak 2-category $\CatBibun$. Therefore, we now give a brief review of smooth 2-groups. For more details, see~\cite{Baez:0307200,Schommer-Pries:0911.2483}.
\begin{definition}
 A \uline{2-group} is a weak monoidal category in which all morphisms are invertible and all objects are weakly invertible.
\end{definition}
\noindent That is, a 2-group is a category $\CCC$ endowed with a unit $\sfe$ and a bifunctor $\otimes:\CCC\times \CCC\rightarrow \CCC$, which satisfies $\sfe\otimes a\cong a$ and $a\otimes \sfe \cong a$ as well as $(a\otimes b)\otimes c\cong a\otimes (b\otimes c)$, where the isomorphisms are given by the left- and right-unitors $\sfl_a$, $\sfr_a$ and the associator $\sfa$. These have to satisfy the usual coherence axioms, cf.\ e.g.~\cite{Schommer-Pries:0911.2483}.

To define smooth 2-groups, we now internalize 2-groups in the weak 2-category $\CatBibun$, see also appendix~\ref{app:A} for a brief review of 2-group objects in a weak 2-category.
\begin{definition}
 A \uline{smooth 2-group} in the sense of~\cite{Schommer-Pries:0911.2483} is a 2-group object in $\CatBibun$.
\end{definition}
\noindent As shown in~\cite{Zhu:0609420}, this definition is equivalent to the canonical definition of a Lie 2-group in terms of simplicial manifolds.\footnote{Here, a Lie 2-group is a Kan complex with one 0-simplex and unique horn fillers for $n$-simplices with $n\geq 3$, cf.~\cite{Henriques:2006aa}.}

Explicitly, a smooth 2-group is given by a Lie groupoid $\CG$ together with a bibundle $\sfm:\CG\times \CG\rightarrow \CG$ and a bibundle $\sfid: (*\rightrightarrows *)\rightarrow \CG$ as well as bibundle morphisms $\sfa$, $\sfl$ and $\sfr$. The bibundle morphisms have to satisfy certain coherence axioms, cf.\ appendix~\ref{app:A}.

This definition of smooth 2-groups subsumes a large number of other notions, as explained in detail in~\cite{Schommer-Pries:0911.2483}. Here, we just summarize the most important examples. First, a Lie group $\sG$, regarded as a Lie groupoid $\sG\rightrightarrows *$ is a smooth 2-group. The monoidal product is the group product in $\sG$, which is promoted to a functor of Lie groupoids and then bundlized. Second,  crossed modules of Lie groups $\sH\xrightarrow{~\dpar~}\sG$ give rise to strict Lie 2-groups, which are special smooth 2-groups, as follows. Consider the groupoid $\sG\times \sH\rightrightarrows \sG$, with structure maps
\begin{subequations}\label{eq:strict_2_group_maps}
\begin{equation}
 \sfs(g,h):=g~,~~~\sft(g,h):=\dpar(h)g\eand \id(g):=(g,\unit_\sH)~.
\end{equation}
Composition of morphisms is defined by
\begin{equation}
 (\dpar(h)g,h')\circ (g,h)=(g,h'h)
\end{equation}
and the tensor product on morphisms is given by the semidirect group action on $\sG\ltimes \sH$,
\begin{equation}
 (g,h)\otimes (g',h'):=(gg',h(g\acton h'))\eand g\otimes g':=gg'~,
\end{equation}
\end{subequations}
where $g,g'\in \sG$, $h,h'\in \sH$ and $\acton:\sG\times \sH\rightarrow \sH$ is the action in the crossed module of Lie groups. One can even show categorical equivalence between crossed modules of Lie groups and strict Lie 2-groups~\cite{Baez:0307200}.

Finally, weak Lie 2-groups, i.e.\ weak 2-groups internal to $\CatManCat$ are also examples of smooth 2-groups.

\subsection{Categories internal to \texorpdfstring{$\CatBibun$}{Bibun}}

In order to define principal 2-bundles with smooth structure 2-groups, we will need the notion of a category internal to $\CatBibun$. Recall that a category internal to a category with pullbacks $\CCC$ is a pair $\CD=(\CD_0,\CD_1)$ of objects in $\CCC$ together with source, target, identity and multiplication morphisms in $\CCC$ such that the usual compatibility conditions between these structure maps for categories hold. Fully analogously, one defines internal functors and internal natural transformations.

The concept of an internal category has been weakened in the past to allow for categories internal to strict 2-categories~\cite{Ferreira:0604549}. Here, we need a slight extension to weak 2-categories to define categories internal to $\CatBibun$.

A more technical issue is that of pullbacks which do not all exist in $\CatBibun$, similarly to the case of $\CatMan$. We can circumvent this problem by introducing the notion of transversality.
\begin{definition}[{\cite[Def.\ 28]{Schommer-Pries:0911.2483}}]
 Let $\CH_{1,2}$ and $\CG$ be Lie groupoids and $B_{1,2}:\CH_{1,2}\rightarrow \CG$ be principal left-bibundles. Then $B_{1,2}$ are \uline{transverse}, if the maps $B_{1,2}\rightarrow \CG_1$ are transverse maps.\footnote{Recall that two maps $f:X\rightarrow Z$ and $g:X\rightarrow Z$ are transverse, if the sum of the pushforwards of $T_pX$ along $f$ and $T_pY$ along $g$ amounts to the full tangent space $T_{f(p)}Z=T_{f(p)}Z$ for all $p,q$ with $f(p)=g(q)$.}
\end{definition}
\noindent We then have the following proposition.
\begin{proposition}[{\cite[Prop.\ 31]{Schommer-Pries:0911.2483}}]
  Let $\CH_{1,2}$ and $\CG$ be Lie groupoids and $B_{1,2}:\CH_{1,2}\rightarrow \CG$ be transverse principal left-bibundles. Then the pullback $\CH_{1}\times_\CG\CH_2$ exists in $\CatBibun$.
\end{proposition}
\noindent With this notion, we are now ready to define categories internal to $\CatBibun$.

\begin{definition}
 A \uline{category $\CCC$ internal to $\CatBibun$} is a pair of Lie groupoids $\CCC_0$ and $\CCC_1$ together with bibundles 
 \begin{equation}
  \sfs,\sft: \CCC_1\rightrightarrows \CCC_0~,~~~\id:\CCC_0\rightarrow \CCC_1~,~~~B_c:\CCC_1\times^{\sfs,\sft}_{\CCC_0}\CCC_1\rightarrow \CCC_1~,
 \end{equation}
 called the source, target, identity and composition morphisms, respectively. We demand that $\sfs$ and $\sft$ are transverse, which guarantees the existence of the pullback $\CCC_1\times^{\sfs,\sft}_{\CCC_0}\CCC_1$. The following diagrams are required to be commutative:
 \begin{equation}
 \begin{aligned}
     \xymatrixcolsep{3pc}
     \xymatrixrowsep{3pc}
     \myxymatrix{
    \CCC_1 \ar@{->}[d]_{\sft} & \ar@{->}[l]_{\pr_1~~~} \CCC_1\times_{\CCC_0}\CCC_1 \ar@{->}[r]^{~~~\pr_2} \ar@{->}[d]^{B_c}& \CCC_1 \ar@{->}[d]_{\sfs}\\
    \CCC_0 & \ar@{->}[l]_{\sft} \CCC_1\ar@{->}[r]^{\sfs} & \CCC_0
    }\hspace{1cm}
     \xymatrixcolsep{3pc}
     \xymatrixrowsep{3pc}
  \myxymatrix{
    \CCC_0 \ar@{->}[r]^{\id} \ar@{->}[dr]_{1} & \CCC_1 \ar@{->}[d]^{\sfs,\sft}\\
    & \CCC_0
    }\hspace{1cm}
 \end{aligned}
 \end{equation}
 We also have bibundle isomorphisms $\sfa$, $\sfl$ and $\sfr$ defined in the commutative diagrams
 \begin{equation}
  \xymatrixcolsep{0pc}
  \xymatrixrowsep{2pc}
  \myxymatrix{
  & \CCC_1\times^{\sfs,\sft}_{\CCC_0}\CCC_1\times^{\sfs,\sft}_{\CCC_0}\CCC_1 \ar@{->}[dl]_{B_c\times 1} \ar@{->}[dr]^{1\times B_c} & \\
  \CCC_1\times^{\sfs,\sft}_{\CCC_0}\CCC_1 \ar@{=>}[rr]^{\sfa} \ar@{->}[dr]_{B_c}& & \CCC_1\times^{\sfs,\sft}_{\CCC_0}\CCC_1 \ar@{->}[dl]^{B_c}\\ & \CCC_1  }\hspace{.2cm}
  \xymatrixcolsep{2pc}
  \xymatrixrowsep{4pc}
  \myxymatrix{
    \CCC_0\times_{\CCC_0}\CCC_1 \ar@{->}[r]^-{\id\times 1} \ar@{->}[dr]_{\pr_2} & \CCC_1\times_{\CCC_0}\times \CCC_1 \ar@{->}[d]^{B_c} \ar@{}[rd]^(.1){}="a"^(.5){}="b" \ar@{=>} "a";"b"_{\sfr} \ar@{}[ld]^(.1){}="c"^(.5){}="d" \ar@{=>} "c";"d"^{\sfl} & \CCC_1\times_{\CCC_0}\CCC_0 \ar@{->}[l]_-{1\times\id} \ar@{->}[dl]^{\pr_1} \\
    & \CCC_1 &
    }
 \end{equation}
 called the \uline{associator}, the \uline{left-} and \uline{right-unitors}, respectively. Coherence of the associator and the unitors amounts to the (internal) pentagon identity, 
 \begin{subequations}\label{diag:int_cat_pentagon_triangle}
    \begin{equation}
      \xymatrixcolsep{-3.1pc}
      \xymatrixrowsep{3pc}
     \myxymatrix{
      & \big[B_c\otimes(1\times B_c)\big]\otimes (B_c\times 1\times 1) \ar@{=>}[dr]^{~(\sfa\otimes 1)\circ \cong}\\
      B_c\otimes\big[(B_c\times 1)\otimes (B_c\times 1\times 1)\big] \ar@{=>}[ur]^{(\sfa\otimes 1)\circ \cong~} \ar@{=>}[d]^{1\otimes (\sfa\times 1)}&  & B_c\otimes\big[(1\times B_c)\otimes (1\times 1\times B_c)\big]\\
      B_c\otimes \big[(B_c\times 1)\otimes (1\times B_c \times 1)\big]\ar@{=>}[rr]^{(\sfa\otimes 1)\circ \cong} & & \big[B_c\otimes (1\times B_c)\big]\otimes (1\times B_c\times 1) \ar@{=>}[u]^{(1\otimes (1\times \sfa))\circ \cong}
    }
    \end{equation}
 as well as the (internal) triangle identity,
    \begin{equation}
     \xymatrixcolsep{2.8pc}\myxymatrix{
      \big[B_c\otimes (B_c\times 1)\big]\otimes (1\times \id\times 1) \ar@{=>}[rr]^{(\sfa\times 1)\otimes 1} \ar@{=>}[dr]_{(1\otimes (\sfr\times 1))\circ \cong~~~~} & & \big[B_c\otimes (1\times B_c)\big]\otimes (1\times \id \times 1) \ar@{=>}[dl]^{~~~~(1\otimes (1\times \sfl))\circ \cong}\\
      & B_c
    }
    \end{equation}
 \end{subequations}
  In the above diagrams, we suppressed arrows for the isomorphisms $\cong$ between bibundles arising from the non-associativity of horizontal or bibundle composition $\otimes$ in $\CatBibun$.
\end{definition}

\noindent Analogously, we now define internal functors. 
\begin{definition}
 Given two categories $\CCC$ and $\CCD$ internal to $\CatBibun$, an \uline{internal functor} $\Phi:\CCC\rightarrow \CCD$ consists of bibundles $\Phi_{0,1}:\CCC_{0,1}\rightarrow \CCD_{0,1}$ and bibundle isomorphisms $\Phi_{2,c}$ and $\Phi_{2,\id}$ such that the following diagrams (2-)commute:
 \begin{equation}\label{diag:int_cat_functor}
     \xymatrixcolsep{3pc}
     \xymatrixrowsep{3pc}
     \myxymatrix{
    \CCC_1 \ar@{->}[d]_{\Phi_1} \ar@{->}[r]_{\sfs} & \CCC_0\ar@{->}[d]_{\Phi_0}& \ar@{->}[l]^{\sft} \CCC_1 \ar@{->}[d]_{\Phi_1}\\
    \CCD_1 \ar@{->}[r]_{\sfs} & \CCD_0 & \ar@{->}[l]^{\sft} \CCD_0
    }\hspace{.5cm}
     \xymatrixcolsep{3pc}
     \xymatrixrowsep{3pc}
  \myxymatrix{
    \CCC_1\times_{\CCC_0}\CCC_1 \ar@{->}[r]^-{B_{c}} \ar@{->}[d]_{\Phi_1\times \Phi_1} & \CCC_1 \ar@{->}[d]^{\Phi_1}\\
    \CCD_1\times_{\CCD_0}\CCD_1 \ar@{->}[r]^-{B_{c}} \ar@{=>}[ur]^{\Phi_{2,c}} & \CCD_1
    }\hspace{.5cm}
     \xymatrixcolsep{3pc}
     \xymatrixrowsep{3pc}
  \myxymatrix{
    \CCC_0 \ar@{->}[r]^{\id} \ar@{->}[d]_{\Phi_0} & \CCC_1 \ar@{->}[d]^{\Phi_1}\\
    \CCD_0 \ar@{->}[r]^{\id} \ar@{=>}[ur]^{\Phi_{2,\id}} & \CCD_1
    }
 \end{equation}
 The bibundle morphisms have to satisfy coherence axioms which amount to the following commutative diagrams\footnote{Note that in these diagrams, the structure 1- and 2-morphisms in $\CCC$ and $\CCD$ are labeled by the same symbols.}:
\begin{subequations}\label{diag:int_functor_coherence}
\begin{equation}\label{diag:int_functor_coherence_a}
\xymatrixcolsep{-3pc}
\xymatrixrowsep{2pc}
\myxymatrix{
& B_c\otimes\big[(\Phi_1\times\Phi_1)\otimes(B_c\times 1)\ar@{=>}[dr]^{\Phi_{2,c}\circ \cong}\big] & \\
  B_c\otimes\big[(B_c\times 1)\otimes(\Phi_1\times\Phi_1\times \Phi_1)\big]\ar@{=>}[ur]^{1\otimes(\Phi_{2,c}\times 1)~~}  \ar@{=>}[d]_{ \sfa\circ \cong} & & \Phi_1\otimes (B_c\times (B_c\times 1)) \ar@{=>}[d]^{1\otimes \sfa}\\
 \big[B_c\otimes(1\times B_c)\big]\otimes(\Phi_1\times\Phi_1\times \Phi_1)\ar@{=>}[dr]_{(1\otimes (1\times\Phi_{2,c}))\circ \cong~~~~} &  & \Phi_1\otimes (B_c\times (1\times B_c))\\
 & B_c\otimes\big[(\Phi_1\times\Phi_1)\otimes(1\times B_c)\big] \ar@{=>}[ur]_{~~~\Phi_{2,c}\circ \cong} & }
\end{equation}
and
\begin{equation}
\xymatrixcolsep{-6pc}
\myxymatrix{
&  B_c\otimes\big[(\Phi_1\times \Phi_1)\otimes\big[(1\times \id)\otimes (1,\sfs)\big]\big] \ar@{=>}[dr]^{\Phi_{2,c}\circ \cong}  & \\
 B_c\otimes\big[(1\times \id)\otimes\big[(\Phi_1\times \Phi_0)\otimes (1,\sfs)\big]\big]\ar@{=>}[ur]^{1\otimes(1\times \Phi_{2,\id})~~~}  \ar@{=>}[dr]_{\sfr} & 
 &  \big[\Phi_1\otimes B_c\big]\otimes\big[(1\times \id)\otimes (1,\sfs)\big] \ar@{=>}[dl]^{~~(1\otimes \sfr)\circ \cong} \\
&  \Phi_1  &\\
 B_c\otimes\big[(\id\times 1)\otimes\big[(\Phi_0\times \Phi_1)\otimes (\sft,1)\big]\big] \ar@{=>}[ur]^{\sfl} \ar@{=>}[dr]_{1\otimes(\Phi_{2,\id}\times 1)~~~~} & 
&  \big[\Phi_1\otimes B_c\big]\otimes\big[(\id\times 1)\otimes (\sft,1)\big] \ar@{=>}[lu]_{~(1\otimes \sfl)\circ \cong}\\
& B_c\otimes\big[(\Phi_1\times \Phi_1)\otimes\big[(\id\times 1)\otimes (\sft,1)\big]\big] \ar@{=>}[ur]_{~~\Phi_{2,c}\circ \cong}  &
}
\end{equation}
\end{subequations}
where we again suppressed additional arrows for isomorphisms arising from the non-asso\-cia\-tivity of horizontal composition in $\CatBibun$. Moreover, we write $(B_1,B_2)$ for the morphism $(B_1\times B_2)\circ \Delta$, where $\Delta$ is the diagonal morphism $\Delta:\CG\rightarrow \CG\times \CG$. The first diagram contains bibundles from $\CCC_1\times_{\CCC_0}\CCC_1\times_{\CCC_0}\CCC_1$ to $\CCD_1$, while the second diagram contains bibundles from $\CCC_1$ to $\CCD_1$.
\end{definition}
\noindent And we finish with internal natural transformations.
\begin{definition}\label{def:internal_natural_trafo}
 Given two internal functors $\Phi$ and $\Psi$ between categories $\CCC$ and $\CCD$ internal to $\CatBibun$, a natural transformation $\beta:\Phi\Rightarrow \Psi$ consists of a bibundle $\CCC_0\rightarrow \CCD_1$ together with a bibundle isomorphism $\beta_2$ rendering the diagrams
 \begin{equation}\label{diag:int_natural_trafo_def}
  \xymatrixcolsep{3pc}
  \xymatrixrowsep{3pc}
  \myxymatrix{
    & \CCC_0 \ar@{->}[dl]_{\Psi_0} \ar@{->}[d]_{\beta_1} \ar@{->}[dr]^{\Phi_0} & \\
    \CCD_0 & \CCD_1 \ar@{->}[l]^{\sft} \ar@{->}[r]_{\sfs} & \CCD_0 
    }\hspace{2cm}
  \xymatrixcolsep{3pc}
  \xymatrixrowsep{3pc}
  \myxymatrix{
    \CCC_1 \ar@{->}[d]_{(\beta_1\otimes \sft,\Phi_1)} \ar@{->}[r]^-{(\Psi_1,\beta_1\otimes\sfs)} & \CCD_1\times \CCD_1 \ar@{->}[d]^{B_c}\\
    \CCD_1\times \CCD_1\ar@{=>}[ur]^{\beta_2} \ar@{->}[r]^{B_c}& \CCD_1
    }
 \end{equation}
 (2-)commutative. In addition, we have coherence rules amounting to the commutative diagrams
 \begin{subequations}\label{diag:int_natural_trafo_coherence}
 \begin{equation}\label{diag:int_natural_trafo_coherence_a}
    \xymatrixcolsep{-14pc}
    \myxymatrix{
    & B_c\otimes\big[(B_c\times 1)\otimes \big[(\Psi_1\times \beta_1\times \Phi_1)\otimes (1\times(\sft,1))\big]\big] \ar@{=>}[dl]_{ \sfa\circ \cong}   & \\
    \big[B_c\otimes(1\times B_c)\big]\otimes \big[(\Psi_1\times \beta_1\times \Phi_1)\otimes (1\times(\sft,1))\big] \ar@{=>}[dd]_{\cong \circ(1\otimes(1\times\beta_2))\circ \cong}   &  & \\
    & & B_c\otimes \big[(B_c\times 1)\otimes\big[(\beta_1\times \Phi_1\times \Phi_1)\otimes((\sft,1)\times 1)\big]\big]  \ar@{=>}[dd]^{\sfa\circ \cong} \ar@/_2.0pc/@{=>}[uul]_{1\otimes(\beta_2\times 1)}\\
 \big[B_c\otimes(1\times B_c)\big]\otimes \big[(\Psi_1\times \Psi_1\times \beta_1)\otimes (1\times(1,\sfs))\big] \ar@{=>}[dd]_{ \sfa^{-1}}& & \\
 & & \big[B_c\otimes (1\times B_c)\big]\otimes\big[(\beta_1\times \Phi_1\times \Phi_1)\otimes((\sft,1)\times 1)\big]   \ar@{=>}[dd]^{(1\otimes(1\times\Phi_{2,c}))\circ \cong}\\
\big[B_c\otimes(B_c\times 1)\big]\otimes \big[(\Psi_1\times \Psi_1\times \beta_1)\otimes (1\times(1,\sfs))\big]\ar@/_2.0pc/@{=>}[ddr]_{(1\otimes(\Psi_{2,c}\times 1))\circ \cong}  
   &  & \\
   & & B_c\otimes \big[(\beta_1\times \Phi_1)\otimes \big[(1\times B_c)\otimes ((\sft,1)\times 1)\big]\big]\ar@{=>}[dl]^{\beta_2}\\
 & B_c\otimes \big[(\Psi_1\times \beta_1)\otimes \big[(B_c\times 1)\otimes (1,(1,\sfs))\big]\big] 
 }
 \end{equation}
 \begin{equation}
\myxymatrix{
& \beta_1 \ar@{=>}[rd]^{~~~(\sfr\circ \cong)^{-1}} & \\
  B_c\otimes \big[(\id\times 1)\otimes (\Psi_0,\beta_1)\big]\ar@{=>}[ru]^{\sfl\circ \cong~~~} \ar@{=>}[d]_{1\otimes(\Psi_{2,\id}\times 1)} &   & B_c\otimes\big[(1\times \id)\otimes(\beta_1,\Phi_0)\big]\ar@{=>}[d]^{1\otimes(1\times\Phi_{2,\id})}  \\
  B_c\otimes\big[(\Psi_1\times 1)\otimes(\id,\beta_1)\big]  & & B_c\otimes\big[(1\times \Phi_1)\otimes(\beta_1,\id)\big]\ar@{=>}[ll]^{\beta_2\otimes \id_{\id}}
  }
 \end{equation}
 \end{subequations}
 The first diagram describes isomorphisms between bibundles from $\CCC_1\times_{\CCC_0}\CCC_1$ to $\CCD_1$ and on this Lie groupoid we have $(1\times (\sft,1))=((1,\sfs)\times 1)$. The second diagram describes isomorphisms between bibundles from $\CCC_0$ to $\CCD_1$ and involves the bibundle isomorphism 
 \begin{equation}
  \beta_2\otimes \id_{\id}:B_c\otimes((\beta_1\otimes \sft,\Phi_1)\otimes \id) \Rightarrow B_c\otimes ((\Psi_1,\beta_1\otimes \sfs)\otimes \id)~.
 \end{equation}
\end{definition}

\section{Smooth 2-group bundles} \label{Section_Smooth_2_group}

\subsection{Ordinary principal bundles}

Recall that given a smooth manifold $X\in\CatMan$, the (generalized) bundles over $X$ are objects in the slice category\footnote{cf.\ appendix~\ref{app:B}} $\CatMan/X$. That is, a generalized bundle is a smooth manifold $P$ with a smooth morphism $P\rightarrow X$.

To obtain a principal bundle $P$ over $X$ with structure Lie group $\sG$, we have to demand that there is a principal group action of $\sG$ on $P$ and that the bundle is locally trivial with typical fiber $\sG$. The first condition is implemented as follows. We switch from $\sG$ to the trivial bundle $\sG_X=(\sG\times X\rightarrow X)$, which is a group object in $\CatMan/X$. We can then demand that $P$ is a $\sG_X$-object $\pi:P\rightarrow X$ in $\CatMan/X$. 

To implement the second condition, we define a cover of $X$ as a smooth manifold $Y$ together with a surjective submersion $\kappa\in\CatMan(Y,X)$. While not all pullbacks exist in $\CatMan$, those along surjective submersions do. For simplicity and for reasons of familiarity, let us restrict ourselves to ordinary covers $\kappa:U\twoheadrightarrow X$ given by a disjoint union of patches, $U:=\djunion_i U_i$. We then demand that $\kappa^*P$ is $\sG$-equivariantly diffeomorphic to the bundle $U\times \sG\rightarrow U\in \CatMan/X$.

We will also need a description of the principal bundle $P$ in terms of descent data or transition functions. For this, we use the $\sG$-equivariant diffeomorphism $\rho_i:U_i\times \sG\rightarrow P|_{U_i}$ to define a transition functions. Note that the diffeomorphism is of the form $\rho_i(p)=(\pi(p),g_i(p))$ for $p\in \pi^{-1}(U_i)$. Then the expression $g_{ij}(p):=g^{-1}_i(p)g_j(p)$ for $p\in \pi^{-1}(U_i\cap U_j)$ depends only on $\pi(p)$ since $g^{-1}_i(hp)g_j(hp)=g^{-1}_i(p)h^{-1}hg_j(p)=g^{-1}_i(p)g_j(p)$. We thus obtain a function $g_{ij}:U_i\cap U_j\rightarrow \sG$, which satisfies the condition $g_{ij}g_{jk}=g_{ik}$ on triple overlaps $U_i\cap U_j\cap U_k\neq \varnothing$. The $(g_{ij})$ thus form a \v Cech 1-cocycle with respect to the cover $U$.

Similarly, one readily shows that diffeomorphic principal bundles $P$ and $P'$ subordinate to the same cover $U$ are described by transition functions $(g_{ij})$ and $(g'_{ij})$ which are related by $g_{ij}=\gamma_i g'_{ij}\gamma_j$ for some local smooth functions $\gamma_i:U_i\rightarrow \sG$. The $(\gamma_i)$ form the \v Cech coboundaries linking the \v Cech cocycles $(g_{ij})$ and $(g'_{ij})$.

Alternatively, one can regard the principal bundle $P$ as a functor from the \v Cech groupoid $\djunion U^{[2]}\rightrightarrows U$ with $U^{[2]}:=\sqcup_{i,j}U_i\cap U_j$ to the Lie groupoid $\sB\sG=(\sG\rightrightarrows *)$. One readily sees that this functor is encoded in a \v Cech 1-cocycle $(g_{ij})$:
\begin{equation}
  \xymatrixcolsep{5pc}
\myxymatrix{
 \sG \ar@<-.5ex>[d] \ar@<.5ex>[d] & U^{[2]}\ar@{->}[l]_{(g_{ij})}  \ar@<-.5ex>[d] \ar@<.5ex>[d] \\
 {*} & U\ar@{->}[l]_{*} 
}
\end{equation}
Moreover, two functors corresponding to diffeomorphic principal bundles are connected by a natural isomorphism, which in turn gives rise to a \v Cech coboundary.

\subsection{Definition of smooth 2-group bundles}

Let us now generalize the above discussion to the categorified setting. This yields higher principal bundles as special kinds of stacks, which were already defined in~\cite{Schommer-Pries:0911.2483}, and we recall the relevant definitions in the following. For a related approach, see also~\cite{Bartels:2004aa}.

Note that a 2-space is a category internal to $\CatMan$ and here, we restrict our attention to Lie groupoids, i.e.\ groupoids internal to $\CatMan$. The 2-bundles over a 2-space $\CX$ are then simply elements of (a subcategory of) the slice 2-category $\CatBibun/\CX$, cf.~\cite{Bartels:2004aa}.

Given a smooth 2-group $\CG$, we can trivially regard it as a 2-group object $\CG_\CX$ in $\CatBibun/\CX$ as follows:
\begin{equation}
  \xymatrixcolsep{2pc}
  \xymatrixrowsep{2pc}
  \myxymatrix{ \CX_1\ar@<-.5ex>[dr] \ar@<.5ex>[dr] & & \ar@{->}[dl]_{\tau} B\ar@{->>}[dr]^{\sigma} & & \ar@<-.5ex>[dl] \ar@<.5ex>[dl] \CG_1\times\CX_1\\
  & \CX_0 &  & \CG_0\times \CX_0
    }
\end{equation}
We then define:
\begin{definition}
 Given a smooth 2-group $\CG$, a \uline{smooth $\CG$-stack} is a $\CG$-object in $\CatBibun$.
\end{definition}
We also define $\CG$-stacks over other smooth stacks $\CX$, which are the objects of $\CatBibun$:
\begin{definition}
 A \uline{smooth 2-group over a smooth stack $\CX$} is a 2-group object in $\CatBibun/\CX$. Given a smooth 2-group $\CG_\CX$ over a smooth stack $\CX$, a \uline{smooth $\CG_\CX$-stack over $\CX$} is a $\CG_\CX$-object in $\CatBibun/\CX$.
\end{definition}

Finally, let us impose the condition of local triviality to arrive at higher principal bundles. To this end, we need to introduce covers and discuss pull-backs to the patches in the covers. It will be sufficient for us to work with covering bibundles arising from a bundlization of 2-covers as defined in~\cite{Baez:2004in}. For a more general perspective, see~\cite{Schommer-Pries:0911.2483}. Since the 2-spaces we want to cover are Lie groupoids, we demand that our cover is also a groupoid $\CU=(\CU_1\rightrightarrows \CU_0)$ internal to $\CatMan$, together with a functor $\tau:\CU\rightarrow \CX$ such that the contained smooth maps $\CU_1\rightarrow \CX_1$ and $\CU_0\rightarrow \CU_1$ are surjective submersions. The bundlization of such a 2-cover gives rise to the bibundle
\begin{equation}
  \xymatrixcolsep{2pc}
  \xymatrixrowsep{2pc}
  \myxymatrix{ \CX_1\ar@<-.5ex>[dr] \ar@<.5ex>[dr] & & \ar@{->}[dl]_{\sft} \CU_0\times^{\tau_s,\sfs}_{\CX_0}\CX_1\ar@{->>}[dr]^{\pi} & & \ar@<-.5ex>[dl] \ar@<.5ex>[dl] \CU_1\\
  & \CX_0 &  & \CU_0
    }
\end{equation}
cf.\ section~\ref{ssec:smooth-2-groups}. Pullbacks exist for surjective submersions, and thus they exist along the corresponding bundlizations\footnote{Using 2-covers saves us the discussion of transversality conditions for bibundles required for pullbacks to exist. For details, see e.g.~\cite{Bartels:2004aa}.}. 

\begin{definition}\label{def:principal_bundle}
 Given a smooth $\CG$-group over a stack $\CX$, a \uline{principal $\CG$-bundle over $\CX$} is a smooth $\CG$-stack $\CP$ over $\CX$ such that there exists a covering bibundle $\kappa:\CU\rightarrow \CX$ with $\kappa^*(\CP)$ being $\CG$-equivariantly equivalent to $\CU\times \CG$ as a smooth $\CG$-stack over $\CU$.
\end{definition}
Altogether we have the following picture:
\begin{equation}\label{diag:def_principal_bundle}
 \xymatrixcolsep{3pc}
    \myxymatrix{
    \CU_1\times \CG_1 \ar@<-.5ex>[d] \ar@<.5ex>[d]&  & \kappa^*\CP_1 \ar@<-.5ex>[d] \ar@<.5ex>[d]& & \CP_1 \ar@<-.5ex>[d] \ar@<.5ex>[d] & \CG_1\times \CX_1\ar@<-.5ex>[d] \ar@<.5ex>[d]\\
    \CU_0\times \CG_0 & B_{{\rm eq}} \ar@{->>}[r] \ar@{->>}[l] & \kappa^*\CP_0 & B_{\kappa^*} \ar@{->>}[r] \ar@{->>}[l]& \CP_0\ar@{=>}[lldd]_\eta & \CG_0\times \CX_0\\
     & B_{\CU\times \CG} \ar[dr] \ar@{->>}[ul]& B_{\kappa^* \CP} \ar[d] \ar@{->>}[u]& & B_\CP \ar[d] \ar@{->>}[u] & \CG_\CX\ar[dl] \ar@{->>}[u]\\
     & \CU_1 \ar@<-.5ex>[r] \ar@<.5ex>[r] & \CU_0 & B_\kappa \ar@{->>}[r] \ar@{->>}[l]& \CX_0 & \CX_1 \ar@<-.5ex>[l] \ar@<.5ex>[l]
    }
\end{equation}
where $B_\CP$ is a $\CG_\CX$-object in $\CatBibun/\CX$, $\eta$ is a bibundle isomorphism, $B_{\rm eq}$ is a bibundle equivalence and $B_{\CU\times \CG}\otimes B_{{\rm eq}}\cong B_{\kappa^*\CP}$.

Let us work through two examples in somewhat more detail: ordinary principal bundles and principal 2-bundles over a manifold $X$, where the structure 2-group is a crossed module of Lie groups. 

In the first case, consider a principal bundle $\pi:P\rightarrow X$ with structure Lie group $\sG$ over a manifold $X$ with cover $\kappa:U\twoheadrightarrow X$. We have an isomorphism $\rho_i:\djunion_i U_i\times \sG \rightarrow \kappa^* P$ such that $\pi\circ\rho_i$ is the obvious projection. To regard these as principal bundles in the sense of definition~\ref{def:principal_bundle}, we first trivially extend the group object $\sG_X=(\sG\times X\rightarrow X)$ to a 2-group object over a Lie groupoid, by promoting $\sG\times X$ and $X$ to discrete categories $\CG=(\sG\times X\rightrightarrows \sG\times X)$ and $\CX=(X\rightrightarrows X)$. The projection in $\sG_X$ induces an obvious functor between $\CG$ and $\CX$, which we can bundlize to the following smooth 2-group over $\CX$:
\begin{equation}\label{eq:bibundle_group_over}
  \xymatrixcolsep{2pc}
  \xymatrixrowsep{2pc}
  \myxymatrix{ X\ar@<-.5ex>[dr] \ar@<.5ex>[dr] & & \ar@{->}[dl]_{\pr} X\times \sG\ar@{->>}[dr]^{=} & & \ar@<-.5ex>[dl] \ar@<.5ex>[dl] X\times \sG\\
  & X &  & X\times \sG
    }
\end{equation}

To obtain a covering bibundle $B_\kappa$ of $\CX$, we proceed similarly. We trivially extend a cover $\kappa:U\twoheadrightarrow X$ to the discrete 2-cover $(\kappa,\kappa):(U\rightrightarrows U)\twoheadrightarrow (X\rightrightarrows X)$, and bundlize the result:
\begin{equation}\label{eq:bibundle_ex_cover}
  \xymatrixcolsep{2pc}
  \xymatrixrowsep{2pc}
  \myxymatrix{ X\ar@<-.5ex>[dr] \ar@<.5ex>[dr] & & \ar@{->}[dl]_{\kappa} U\ar@{->>}[dr]^{=} & & \ar@<-.5ex>[dl] \ar@<.5ex>[dl] U\\
  & X &  & U
    }
\end{equation}
Similarly, all the other maps are generalized to bibundles by bundlization of the corresponding functors between discrete groupoids and it is obvious how to complete diagram~\eqref{diag:def_principal_bundle}. In particular, $\eta$ is trivial.

In the case of principal 2-bundles over $X$, we choose the strict structure Lie 2-group $\CG=(\sG\ltimes \sH\rightrightarrows \sG)$. Again, we promote $X$ and its cover $U$ to discrete groupoids. We have an obvious functor from the Lie groupoid $\CG\times \CX$ to $\CX$, which we bundlize to 
\begin{equation}
  \xymatrixcolsep{2pc}
  \xymatrixrowsep{2pc}
  \myxymatrix{ X\ar@<-.5ex>[dr] \ar@<.5ex>[dr] & & \ar@{->}[dl]_{\pr} \sG\times X\ar@{->>}[dr]^{=} & & \ar@<-.5ex>[dl] \ar@<.5ex>[dl] (\sG\ltimes \sH)\times X\\
  & X &  & \sG\times X
    }
\end{equation}
As covering bibundle, we choose again~\eqref{eq:bibundle_ex_cover}. Recall that a principal $\CG$ bundle over $X$ can be regarded as a 2-space $\CP$ fibered over $X$, whose pullback along the cover is equivalent to the bundle $\CG\times (U\rightrightarrows U)$, cf.\ e.g.~\cite{Wockel:2008aa}. Bundlization then allows us to fill in all the remaining bibundles of diagram~\eqref{diag:def_principal_bundle}.

\subsection{Cocycle description}\label{ssec:cocycle_description}

Let $\CG$ be a smooth 2-group. One can now derive transition functions by locally trivializing the above description of principal $\CG$-bundles in the usual manner.

Alternatively, we can directly derive a description in terms of generalized \v Cech cocycles. Let $\kappa:\CU\rightarrow \CX$ be a covering bibundle of a stack $\CX$. We can construct the \v Cech groupoid $\check \CCC(\CU)$ of $\CU_1\rightrightarrows \CU_0$ as the obvious category internal to $\CatBibun$. Correspondingly, we construct $\sB \CG$ of the smooth 2-group $\CG$ as a category in $\CatBibun$. We then have the following definition, generalizing the usual \v Cech description of principal fiber bundles.
\begin{definition}
 A \uline{principal $\CG$-bundle} over $\CX$ subordinate to a cover $\CU$ of $\CX$ is a functor internal to $\CatBibun$ from the \v Cech groupoid $\check \CCC(\CU)$ to the delooping $\sB\CG$ of the smooth structure 2-group $\CG$. Two principal $\CG$-bundles over $\CX$ subordinate to a cover $\CU$ are called \uline{equivalent}, if there is a natural isomorphism between their corresponding functors.
\end{definition}
Altogether, we get the following diagram:
\begin{equation}
     \xymatrixcolsep{3pc}
     \xymatrixrowsep{3pc}
    \myxymatrix{
     \CG_1  \rightrightarrows\CG_0\ar@<-.5ex>[d] \ar@<.5ex>[d] &
    \CU_{1,0}\leftleftarrows \CU_{1,1} \ar@<-.5ex>[d]_{\sfs} \ar@<.5ex>[d]^{\sft}
    \ar@/^2pc/[l]^(.1){}="c"^(.5){}="d"
    \ar@/_2pc/[l]^(.1){}="e"^(.5){}="f" \ar@{=>} "d";"f"^{\beta}
    \ar@/^2pc/[l]^{\Phi_1}
    \ar@/_2pc/[l]_{\Psi_1}\\
    {*}\rightrightarrows{*} & \CU_0\leftleftarrows\CU_1 \ar@{->}[l]^{\Phi_0}_{\Psi_0} 
    }
\end{equation}
where $\Phi$ and $\Psi$ are internal functors, $\beta$ is an internal natural isomorphism and the maps $\sfs$ and $\sft$ are bibundles\footnote{In principle, the maps $(\CG_1\rightrightarrows\CG_0)\rightrightarrows (*\rightrightarrows*)$ are also given by bibundles, but since the target is trivial, they collapse to trivial maps.}. 

As a particularly simple example, consider the principal $\CG$-bundle $\Phi^\unit$ whose $\Phi^\unit_1$-com\-ponent is given by the bundlization of the functor mapping all of $\CU_{1,0}$ to $\unit_\CG\in \CG_0$ and all of $\CU_{1,1}$ to $\id_{\unit_\CG}\in \CG_1$.
\begin{definition}\label{def:trivial_principal_G_bundle}
 A \uline{trivial principal $\CG$-bundle} is a principal $\CG$-bundle which is equivalent to the principal $\CG$-bundle $\Phi^\unit$.
\end{definition}

Let us explain how ordinary principal bundles over a manifold $X$ with structure group $\sG$ fit into this definition. If $\CX$ is the discrete groupoid $X\rightrightarrows X$, then we can also choose the cover $\CU$ to be discrete. In this case, the maps $\sfs$ and $\sft$ collapse to smooth maps between $\CU_{1,0}$ and $\CU_{0}$. The \v Cech groupoid $\check \CCC(\CU)$ can be reduced to the \v Cech groupoid of an ordinary cover $\CU_0=U=\sqcup_i U_i$ of $X$ and the composition of compatible elements in $\CU_{1,0}=U^{[2]}=\sqcup_{i,j} U_{i}\cap U_j$ is the (bundlization of) the usual composition of double overlaps. The groupoid $\CG$ is now the discrete Lie groupoid $\sG\rightrightarrows \sG$ and the composition bibundle is simply the bundlization of the multiplication map, trivially lifted to a functor. Given this initial data, the bibundles contained in $\Phi$ and $\Psi$ reduce to smooth maps $(g_{ij}):U^{[2]}\rightarrow \sG$. Their composition with multiplication appearing in the second diagram of~\eqref{diag:int_cat_functor} is encoded in the bibundles
\begin{equation}
  \xymatrixcolsep{2pc}
  \xymatrixrowsep{2pc}
  \myxymatrix{ \sG\ar@<-.5ex>[dr] \ar@<.5ex>[dr] & & \ar@{->}[dl]\sG\times \sG\ar@{->>}[dr] & \sG\times \sG\ar@<-.5ex>[d] \ar@<.5ex>[d] & \ar@{->}[dl]U^{[3]}\ar@{->>}[dr] & & \ar@<-.5ex>[dl] \ar@<.5ex>[dl]U^{[3]}\\
  & \sG & & \sG\times \sG & & U^{[3]} 
    }
\end{equation}
with $U^{[3]}:=\sqcup_{i,j,k} U_{ij}\times_M U_{jk}=\sqcup_{i,j,k} U_i\cap U_j\cap U_k$ and the second bibundle is the bundlization of the group multiplication. These bibundles compose to the bibundle
\begin{equation}
  \xymatrixcolsep{2pc}
  \xymatrixrowsep{2pc}
  \myxymatrix{ \sG\ar@<-.5ex>[dr] \ar@<.5ex>[dr] & & \ar@{->}[dl]U^{[3]}\ar@{->>}[dr] & & \ar@<-.5ex>[dl] \ar@<.5ex>[dl]U^{[3]}\\
  & \sG &  & U^{[3]} 
    }
\end{equation}
Altogether, we recover the usual \v Cech cocycles encoding transition functions of a principal $\sG$-bundle over $X$ subordinate to the cover $U$:
\begin{equation}
 g_{ij}(x)g_{jk}(x)=g_{ik}(x)~,~~~x\in U_i\cap U_j\cap U_k~.
\end{equation}
Analogously, the bibundle morphisms contained in $\beta$ arise from bundlizing smooth maps $(\gamma_i):(U_i)\rightarrow \sG$. If $(g_{ij})$ and $(g'_{ij})$ are the cocycles corresponding to the functors $\Phi$ and $\Psi$, then we have
\begin{equation}
 \gamma_i(x) g_{ij}(x)=g'_{ij}(x) \gamma_j(x)~,
\end{equation}
and the $\gamma_i$ form a \v Cech coboundary.

\subsection{Example: Principal 2-bundles with strict structure 2-group}

As a preparation for discussing principal 2-bundles with smooth 2-groups as their structure 2-groups, let us also go through the example of principal 2-bundles with strict structure 2-group in much detail.

The relevant 2-cover is again derived from the \v Cech groupoid of the underlying manifold $X$ as a category internal to $\CatBibun$ as done in the previous section. The structure 2-group is given by a strict Lie 2-group $\CG=(\sG\ltimes \sH)\rightrightarrows \sG$, which is regarded as a category $\CG\rightrightarrows *$ internal to $\CatBibun$ with the bibundle $B_c$ being the monoidal product in the strict Lie 2-group $\CG$.

Here, the bibundle $\Phi_1$ (and $\Psi_1$) no longer collapses straightforwardly. To simplify the discussion, let us assume that the cover $U$ is sufficiently fine so that $U^{[2]}=\sqcup_{i,j} U_{ij}$ is contractible. Then the bibundle $\Phi_1$ reads as
\begin{equation}\label{eq:bibundle_strict_2_group}
  \xymatrixcolsep{2pc}
  \xymatrixrowsep{2pc}
  \myxymatrix{ \sG\ltimes \sH\ar@<-.5ex>[dr] \ar@<.5ex>[dr] & & \ar@{->}[dl]_{\tau}\Phi_1=U^{[2]}\times \sH\ar@{->>}[dr]^{\sigma} & & \ar@<-.5ex>[dl] \ar@<.5ex>[dl]U^{[2]}\\
  & \sG &  & U^{[2]} 
    }
\end{equation}
where $\sigma$ is the projection. The bibundle $\Phi_1$ is now necessarily a trivial bibundle over $U^{[2]}$ and therefore isomorphic to a bundlization. Instead of using this fact, let us come to this conclusion by explicitly working through the details. 

Note that $\tau$ is fully fixed by its image of elements $(i,j,x,\unit_\sH)\in U^{[2]}\times \sH$, because the left-action fixes the remaining part of $\tau$. In particular,
\begin{equation}
 (i,j,x,h)=\big(\tau(i,j,x,\unit_\sH),h\big)\,(i,j,x,\unit_{\sH})~,
\end{equation}
and thus
\begin{equation}
 \tau(i,j,x,h)=\sft\big(\tau(i,j,x,\unit_\sH),h\big)~.
\end{equation}
We therefore define
\begin{equation}
 g_{ij}(x):=\tau(i,j,x,\unit_\sH)~,
\end{equation}
implying $\tau(i,j,x,h)=\dpar(h)g_{ij}(x)$. Altogether, we see that the bibundle $\Phi_1$ is simply the bundlization of the functor
\begin{equation}
  \xymatrixcolsep{5pc}
  \xymatrixrowsep{2pc}
  \myxymatrix{ \sG\ltimes \sH\ar@<-.5ex>[d] \ar@<.5ex>[d] & U^{[2]}\ar@<-.5ex>[d] \ar@<.5ex>[d] \ar@{->}[l]_{(g_{ij},\unit_\sH)} \\
  \sG  & U^{[2]} \ar@{->}[l]_{(g_{ij})}
    }
\end{equation}
as expected since the bibundle is a trivial bundle over $U^{[2]}$.

Let us now consider the appropriate version of the second diagram in~\eqref{diag:int_cat_functor}, which encodes (weak) compatibility of the internal functor with bibundle composition:
\begin{equation}
  \xymatrixcolsep{2pc}
  \xymatrixrowsep{2pc}
  \myxymatrix{ & & U^{[3]}\times \sH \ar@{->}[dl]_{\tau\otimes \tau} \ar@{->>}[dr]^{}\ar@{=>}[dd]^{\Phi_{2,c}} & \\
  \sG\ltimes \sH \ar@<-.5ex>[r] \ar@<.5ex>[r] & \sG & &U^{[3]} & \ar@<-.5ex>[l] \ar@<.5ex>[l] U^{[3]}\\
  & & U^{[3]}\times \sH \ar@{->}[ul]^{\tau\circ \pr_{13}} \ar@{->>}[ur]&
    }
\end{equation}
with $U^{[3]}=\djunion_{i,j,k} U_{ijk}$. Since the map $\Phi_{2,c}$ is compatible with the principal left-action and the projections $\sigma$, it is fully determined by the function $h:U^{[3]}\rightarrow \sH$ defined implicitly according to
\begin{equation}
 \Phi_{2,c}(i,j,k,x,\unit_\sH)=(i,j,k,x,h^{-1}_{ijk}(x))~,~~~(i,j,k,x,\unit_\sH)\in U^{[3]}\times \sH~,
\end{equation}
where we chose to invert $h_{ijk}$ for consistency with conventions e.g.\ in~\cite{Saemann:2012uq}. The condition that $(\tau\otimes \tau)=\tau\circ \pr_{13}\circ \Phi_{2,c}$ then directly translates into the equation
\begin{equation}\label{eq:strict_2_cocycle_1}
 \dpar(h_{ijk}(x))g_{ij}(x)g_{jk}(x)=g_{ik}(x)~.
\end{equation}
Also, the coherence axiom~\eqref{diag:int_functor_coherence_a} amounts to
\begin{equation}\label{eq:3.21}
\Phi_{ikl}\circ (\Phi_{ijk}\otimes \id_{\Phi_{kl}})=\Phi_{ijl}\circ (\id_{\Phi_{ij}}\otimes\Phi_{jkl})~,
\end{equation}
where the restriction of $\Phi_c:U^{[3]}\times \sH\rightarrow U^{[3]}\times \sH$ to $\Phi_{ijk}:U_i\cap U_j\cap U_k\times \sH\rightarrow U_i\cap U_j\cap U_k\times \sH$ and of $\Phi:(U^{[2]}\rightrightarrows U^{[2]})\rightarrow (\sG\ltimes \sH\rightrightarrows \sG)$ to $\Phi_{ij}:(U_{ij}\rightrightarrows U_{ij})\rightarrow (\sG\ltimes \sH\rightrightarrows \sG)$ appear. Evaluating~\eqref{eq:3.21} on $(i,j,k,l,x,\unit_\sH)$ using the formulas~\eqref{eq:strict_2_group_maps}, we obtain the relation
\begin{equation}\label{eq:strict_2_cocycle_2}
 h_{ikl} h_{ijk}=h_{ijl} (g_{ij}\acton h_{jkl})~.
\end{equation}
Equations~\eqref{eq:strict_2_cocycle_1} and~\eqref{eq:strict_2_cocycle_2} are the usual cocycle relations for a principal 2-bundle with strict structure 2-group.

Given two such cocycles $(g_{ij},h_{ijk})$ and $(g'_{ij},h'_{ijk})$, we can consider internal natural isomorphisms between them, cf.~definition~\ref{def:internal_natural_trafo}. Such an isomorphism $\beta$ is encoded in a bibundle $\beta_1$ from $U\rightrightarrows U$ to $\sG\ltimes \sH\rightrightarrows \sH$ and a bibundle isomorphism $\beta_2$ contained in
\begin{equation}
  \xymatrixcolsep{2pc}
  \xymatrixrowsep{2pc}
  \myxymatrix{ & & B_c\otimes\big((\beta_1\otimes \sft),\Phi_1\big) \ar@{->}[dl] \ar@{->>}[dr]^{}\ar@{=>}[dd]^{\beta_2} & \\
  \sG\ltimes \sH \ar@<-.5ex>[r] \ar@<.5ex>[r] & \sG & &U^{[2]} & \ar@<-.5ex>[l] \ar@<.5ex>[l] U^{[2]}\\
  & & B_c\otimes\big(\Psi_1,(\beta_1\otimes \sfs)\big) \ar@{->}[ul] \ar@{->>}[ur]&
    }
\end{equation}
Here, $B_c$ is the bundlization of the vertical composition functor in the strict 2-group $\sG\ltimes \sH\rightrightarrows \sH$ and we use again the standard notation $(B_1,B_2):=(B_1\times B_2)\otimes \Delta$, where $\Delta:\CG\rightarrow \CG\times \CG$ is the appropriate diagonal bibundle. Following arguments analogous to those given above, the bibundle $\beta_1$ is diffeomorphic to $U\times \sH$ and the map $\tau:U\times \sH\rightarrow \sG$ is fully determined by maps $\gamma_i(x):=\tau(i,x,\unit_\sH)$. Moreover, the bibundles related by the isomorphism $\beta_2$ are isomorphic to $U^{[2]}\times \sH$, and the isomorphism $\beta_2$ is fixed by maps $\chi_{ij}(x):=\beta_2(i,j,x,\unit_\sH)$. The second diagram in~\eqref{diag:int_natural_trafo_def} then immediately yields the equation
\begin{subequations}\label{eq:strict_2_coboundary}
\begin{equation}
   \gamma_i g_{ij}=\dpar(\chi_{ij})g'_{ij}\gamma_j~.
\end{equation}
The commutative diagram~\eqref{diag:int_natural_trafo_coherence_a} simplifies a bit, because all associators are trivial. Evaluating the bibundle isomorphisms at $(i,j,k,x,\unit_\sH)$ in $U^{[3]}\times \sH$, we obtain the relation
\begin{equation}
  \chi_{ik}h'_{ijk}=(\gamma_i\acton h_{ijk})\chi_{ij}(g'_{ij}\acton \chi_{jk})~.
\end{equation}
\end{subequations}
Equations~\eqref{eq:strict_2_coboundary} give the usual coboundary relation for a principal 2-bundle with strict structure 2-group, as found e.g.\ in~\cite{Baez:2004in} or in the conventions close to ours here in~\cite{Saemann:2012uq}.

\section{The string group}

\subsection{General remarks} \label{Section-general_remark_string_grp}

The string group $\sString(n)$ is a 3-connected cover of the spin group $\sSpin(n)$. It fits within the Whitehead tower of the orthogonal group $\sO(n)$. Recall that the Whitehead tower over a space $X$ consists of a sequence of spaces 
\begin{equation}
 *\rightarrow \dots \xrightarrow{~\nu_{i+1}~} X^{(i)}\xrightarrow{~\nu_{i}~}\dots \xrightarrow{~\nu_{3}~} X^{(2)}\xrightarrow{~\nu_{2}~} X^{(1)}\xrightarrow{~\nu_{1}~} X~,
\end{equation}
where the maps $\nu_i$ induce isomorphisms on all homotopy groups in degree $k\geq i$ and $\pi_{j}(X^{(i)})=0$ for $j< i$. In the case of $\sO(n)$, we have
\begin{equation}
 \dots \rightarrow \sString(n)\rightarrow \sSpin(n)\rightarrow \sSpin(n)\rightarrow \sSO(n)\rightarrow \sO(n)~.
\end{equation}

The string group is only defined up to homotopy, and therefore the group structure can only be determined up to $A_\infty$-equivalence. Moreover the smooth structure on the string group is not determined at all. Therefore, there exist various different models and the first geometric model as a topological group was constructed by Stolz~\cite{Stolz:1996:785-800} and Stolz and Teichner in~\cite{Stolz:2004aa}. Because $\pi_1$ and $\pi_3$ of $\sString(n)$ vanish, the string group cannot be modeled by a finite-dimensional Lie group. 

Looking for ways to circumvent this issue, one is naturally led to Lie 2-group models of the string group~\cite{Baez:0307200}. These are Lie 2-groups endowed with a Lie 2-group homomorphism to $\sSpin(n)$, regarded as a Lie 2-group. A first such model was constructed in~\cite{Baez:2005sn}, which is a strict but infinite-dimensional Lie 2-group and differentiates to a strict Lie 2-algebra which is equivalent to the string Lie 2-algebra. Closely related is the construction of~\cite{Henriques:2006aa}, which yields an integration of the string Lie 2-algebra as a simplicial manifold. Moreover, there is an infinite-dimensional model as a strict Lie 2-group~\cite{Nikolaus:2011zg} which was obtained by smoothening the original Stolz--Teichner construction. The model we shall be mostly interested in here is that of Schommer--Pries~\cite{Schommer-Pries:0911.2483}: a group object in $\CatBibun$ which is semistrict but finite dimensional. We believe that this model is best suited for a description of physically interesting solutions to higher gauge theory.

\subsection{Differentiable hypercohomology}\label{ssec:diff-hypercohomology}

A particularly interesting Lie 2-algebra is the string Lie 2-algebra of a compact simple Lie group $\sG$, and we will encounter its explicit form later. This Lie 2-algebra is fully characterized by the Cartan--Killing form on a Lie group, which represents an element of $H^3(\sLie(\sG),\FR)$. In~\cite{Baez:0307200}, the authors showed that Lie 2-groups are classified by a pair of groups $\sG$, $\sH$, with $\sH$ abelian, an action of $\sG$ on $\sH$ by automorphism and an element of $H^3(\sG,\sH)$. It is thus tempting to assume that the string Lie 2-algebra can be integrated to such  classifying data. As shown in~\cite{Baez:0307200}, however, this cannot be done if the underlying topology is to be respected. 

The reason behind this problem is that ordinary group cohomology is not the right framework for this integration. As done in~\cite{Schommer-Pries:0911.2483}, one should rather switch to Segal--Mitchison group cohomology~\cite{Segal:1970aa}, which we briefly review in the following.

Recall that given a simplicial set $S_\bullet=\bigcup_{p=0}^\infty S_p$, we have face and degeneracy maps\footnote{Note that our symbols for these maps differ from another widespread choice.} $\sff^p_i:S_p\rightarrow S_{p-1}$ and $\sfd^p_i:S_p\rightarrow S_{p+1}$, $0\leq i\leq p$. The former induce a coboundary operator on functions on $S_\bullet$, $\delta:\CC^\infty(S_{p-1})\rightarrow \CC^\infty(S_{p})$, via $(\delta f)(s):=\sum_{j=0}^p(-1)^j f(\sff^p_j s)$ for $s\in S_p$.

Given a manifold $M$ together with a good cover $\pi:V_1=\sqcup_i(V_i)\twoheadrightarrow M$, we can define a simplicial set, the {\em nerve of the \v Cech groupoid}, 
as the fibered product\footnote{If $\pi_i:V_i\rightarrow M$ are the restrictions of $\pi$, then the fibered product is defined as $V_i\times_M V_j:=\{(i,j,x)|\pi_i(x)=\pi_j(x)\}$.} 
\begin{equation}
 V_\bullet=\bigsqcup_{p=1}^\infty V^{[p]}=\bigsqcup_{p=1}^\infty \bigsqcup_{i_1,\dots,i_p} V_{i_1}\times_M V_{i_2}  \times_M\dots \times_M V_{i_p}~.
\end{equation}
Sheaf-valued maps on $V^{[p]}$ are called a \v Cech $(p-1)$-cochains. Together with the corresponding simplicial coboundary operator $\delta_{\check{C}}$, they form a complex. \v Cech cohomology with values in the sheaf $\CCS$ is simply the cohomology of that complex.

In many constructions in category theory, and in particular in higher category theory, it is actually more convenient to talk about the nerve of a category than about the category itself. Consider for example the nerve $N(\sB\sG)$ of the groupoid $\sB\sG$, which is the simplicial set $\sG_\bullet=\bigcup_{p=0}^\infty \sG^{\times p}$, where $\sG$ is some Lie group and $\sG^{\times 0}=*$. The face and degeneracy maps are given by
\begin{equation}
 \begin{aligned}
  \sff^1_0(g_1)=\sff^1_1(g_1)&=*~,\\                                
  \sff^p_i(g_1,\dots,g_p)&=\left\{\begin{array}{ll}
				    (g_2,\dots, g_p) & \mbox{if $i=p>1$}~,\\
				    (g_1,\dots, g_{p-1}) & \mbox{if $i=0$, $p>1$}~,\\
				    (g_1,\dots,g_{i-1}g_i,g_{i+1},\dots,g_p) & \mbox{if $0<i<p>1$}~,
                                  \end{array}\right.\\
  \sfd^0_0(*)&=\unit_\sG~,\\
  \sfd^p_i(g_1,\dots,g_p)&=(g_1,\dots,g_{i-1},g_i,g_i,\dots,g_p)~.
 \end{aligned}
\end{equation}
We denote the differential arising as a coboundary operator of this simplicial complex by $\delta_N$. 

To combine this simplicial complex with that arising from the \v Cech groupoid, we need to consider a simplicial cover of $\sG_\bullet$. Our definition of such a cover will come with somewhat more structure than that of~\cite{Schommer-Pries:0911.2483}, cf.~\cite{Brylinski:2001aa}.
\begin{definition}
 A \uline{simplicial cover} $(V_\bullet,I_\bullet)$ of a simplicial manifold $M_\bullet$ is a simplicial set $I_\bullet$ together with a simplicial manifold $V_\bullet$ covering $M_\bullet$ such that for all $j\in I_p$,
 \begin{equation}
  \sff_i (V_{p,j})\subset V_{p-1,\sff_i(j)}\eand \sfd_i (V_{p,j})\subset V_{p+1,\sfd_i(j)} ~,
 \end{equation}
 where $0\leq i\leq p$ or $0\leq i\leq p+1$, respectively, and the face and degeneracy maps are those of $V_\bullet$ and $I_\bullet$.
\end{definition}

Given now an abelian group $\sA$, we can consider the hypercohomology of smooth $\sA$-valued \v Cech cochains on $\sG_\bullet$, where the differentials are induced by the two simplicial structures. We have the following double complex.
\begin{equation}
 \vcenter{\vbox{\xymatrix{
\vdots & \\
\CC^\infty\big(V^{[1]}_3,\sA\big) \ \ar[u]^-{\delta_N} \ \ar[r]^-{\delta_{\check{C}}} & \Ddots\\
\CC^\infty\big(V^{[1]}_2,\sA\big) \ \ar[u]^-{\delta_N} \ \ar[r]^-{\delta_{\check{C}}} & \
\CC^\infty\big(V^{[2]}_2,\sA\big) \ \ar[u]^-{\delta_N} \ar[r]^-{\delta_{\check{C}} }  & \
\Ddots \\
\CC^\infty\big(V^{[1]}_1,\sA\big) \ \ar[u]^-{\delta_N} \ \ar[r]^-{\delta_{\check{C}}} & \
\CC^\infty\big(V^{[2]}_1,\sA\big) \ \ar[u]^-{\delta_N} \ar[r]^-{\delta_{\check{C}} } & \
\CC^\infty\big(V^{[3]}_1,\sA\big) \ \ar[u]^-{\delta_N} \ar[r]^-{\delta_{\check{C}} } & \
\Ddots \\
\CC^\infty\big(V^{[1]}_0,\sA\big) \ \ar[u]^-{\delta_N} \ \ar[r]^-{\delta_{\check{C}}} & \
\CC^\infty\big(V^{[2]}_0,\sA\big) \ \ar[u]^-{\delta_N} \ar[r]^-{\delta_{\check{C}} } &
\CC^\infty\big(V^{[3]}_0,\sA\big) \ \ar[u]^-{\delta_N} \ar[r]^-{\delta_{\check{C}} } &
\CC^\infty\big(V^{[4]}_0,\sA\big) \ \ar[u]^-{\delta_N} \ar[r]^-{\delta_{\check{C}} } &
\dots 
}}}
\end{equation}
Note that $V_0$ covers the point $*$ and therefore the bottom line of the diagram above can be chosen to be trivial. For simplicity, we shall label the $(p,q)$-cochains by $\CC^{p,q}(\sA):=\CC^\infty(V^{[p+1]}_q,\sA)$ in the following. Segal--Mitchison cohomology is now the total cohomology of this double complex. The underlying differential is
\begin{equation}
 \delta_{\rm SM}=\delta_{\check{C}}+(-1)^p\delta_N~~:~~\bigsqcup_{p=0}^{n}\CC^\infty(V^{[p+1]}_{n-p},\sA)=\bigsqcup_{p=0}^{n}\CC^{p,n-p}(\sA)~\rightarrow~ \bigsqcup_{p=0}^{n+1}\CC^{p,n+1-p}(\sA)~,
\end{equation}
where $p$ is the \v Cech degree of the cochain that $\delta_{\rm SM}$ acts on. We shall always work with normalized cocycles, which become trivial if two subsequent arguments are identical.

As an example, consider a representative $\lambda$ of a generator of $H^3_{\rm SM}(\sSpin(n),\sU(1))$. Such an element encodes a model for the string group as shown later. It is given by four smooth maps\footnote{For comparison, $\lambda_1$, $\lambda_2$ and $\lambda_3$, $\delta_h$ and $\delta_v$ in~\cite{Schommer-Pries:0911.2483} correspond to $\lambda^{2,1}$, $\lambda^{1,2}$, $\lambda^{0,3}$, $\delta_{\check C}$ and $\delta_N$, respectively.}
\begin{subequations}\label{eq:cocycle-example}
\begin{equation}
 \lambda=(\lambda^{3,0}=0,\lambda^{2,1},\lambda^{1,2},\lambda^{0,3})~,~~~\lambda^{i,j}\in \CC^{i,j}(\sU(1))~,
\end{equation}
where the cocycle condition $\delta_{\rm SM}\lambda=0$ reads as
\begin{equation}
 0=\delta_{\check{C}}\lambda^{2,1}~,~~~\delta_N\lambda^{2,1}=\delta_{\check{C}}\lambda^{1,2}~,~~~
 \delta_N\lambda^{1,2}=\delta_{\check{C}}\lambda^{0,3}~,~~~\delta_N\lambda^{0,3}=0~.
\end{equation}
\end{subequations}
Evidently, the map $\lambda^{2,1}$ defines an element in $H^2(\sSpin(n),\sU(1))$ and therefore encodes an abelian gerbe over $\sSpin(n)$. 

To conclude, let us briefly show how one can construct a simplicial cover of $\sSpin(n)_\bullet=N(\sB\sSpin(n))$ following~\cite{Brylinski:2001aa}, which is the starting point for constructing an element of $H^3_{\rm SM}(\sSpin(n),\sU(1))$. We focus on the case $n=3$, but our construction readily generalizes to arbitrary $n$. 

\begin{example}\label{ex:cover}
An element $g \in \sSpin(3)\cong\sS\sU(2)$ is parameterized by a real vector $(x^1,x^2,x^3,$ $x^4)$ of length~1 as follows:
   \begin{equation}
      \begin{aligned}
      g=
      \begin{pmatrix}
            x^1+ix^2 & x^3+ix^4\\
            -x^3+ix^4 & x^1-ix^2
      \end{pmatrix}
      ~.
      \end{aligned}
   \end{equation}  
   A convenient cover of $\sSU(2)$ is given by $V_1=V_1^{[1]}=\sqcup_{i\in I_1}V_{1,i}$ with $I_1=\{1,\dots,8\}$ and
 \begin{equation}\label{cover_SU2}
   \begin{aligned}
    & V_{1,1} =  \{g \in \sS\sU(2)|~ x^1\geq 0 \}~,~~ V_{1,2}  =  \{g \in \sS\sU(2)|~ x^1< 0\}~,~~
     V_{1,3}  = \{g \in \sS\sU(2)|~ x^2\geq 0\}~, \\
     &V_{1,4}   =  \{g \in \sS\sU(2)|~ x^2< 0\}~,~~
     V_{1,5}  =  \{g \in \sS\sU(2)|~ x^3\geq 0 \}~,~~
     V_{1,6}  =  \{g \in \sS\sU(2)|~ x^3< 0\}~,\\
     & V_{1,7}  = \{g \in \sS\sU(2)|~ x^4\geq 0\}~, \eand
     V_{1,8}  = \{g \in \sS\sU(2)|~ x^4< 0\}~.
   \end{aligned}
\end{equation}
The index set $I_1$ is now trivially extended to a simplicial set $I_\bullet$ by using multiindices:
\begin{equation}
 I_2=\{(i_1,i_2,i_3)|i_{1,2,3}\in I_1\}~,~~~
 I_3=\{(j_1,j_2,j_3,j_4)|j_{1,2,3,4}\in I_2\}~,~~~\mbox{etc.}
\end{equation}
The actions of the face $\sff^p_i$ and degeneracy maps $\sfd^p_i$ are obvious: the former drop the $i$-th slot, while the latter double the $i$-th slot. Note that the $I_p$ are finite and carry a total order induced by the lexicographic ordering of indices.

The simplicial cover $V_\bullet$ is then obtained from the preimages of the face maps of the nerve of $\sB\sSU(2)$:
\begin{equation}
\begin{aligned}
 V_{2,(i_1,i_2,i_3)}&:=(\sff^2_0)^{-1}(V_{1,i_1})~\cap~ (\sff^2_1)^{-1}(V_{1,i_2})~\cap~ (\sff^2_2)^{-1}(V_{1,i_3})~,\\
 V_{3,(j_1,j_2,j_3,j_4)}&:=(\sff^3_0)^{-1}(V_{2,j_1})~\cap~(\sff^3_1)^{-1}(V_{2,j_2})~\cap~(\sff^3_2)^{-1}(V_{2,j_3})~\cap~(\sff^3_3)^{-1}(V_{2,j_4})~,\\
\end{aligned}
\end{equation}
etc.\ with the obvious face and degeneracy maps.

The lexicographic ordering of indices allows us to introduce a section $\phi$ of $\pi:V_\bullet \rightarrow N(\sB\sSpin(3))$. In particular, $\phi_1(g)$ is the element $v\in V_{1,i}$ with $\pi(v)=g$ and $i$ as small as possible.
\end{example}

\subsection{The string group model of Schommer--Pries} \label{Schommer-Pries_model}

In~\cite{Schommer-Pries:0911.2483}, Schommer--Pries constructed a smooth 2-group model of the string group, and we briefly recall this construction in the following. First, we need to generalize the extension of Lie groups by other Lie groups to the categorified setting, as done in~\cite[Def.\ 75]{Schommer-Pries:0911.2483}:
\begin{definition}
 An \uline{extension} of a smooth 2-group $\CG$ by a smooth 2-group $\CA$ consists of a smooth 2-group $\CE$ together with homomorphisms $f:\CA\rightarrow \CE$, $g:\CE\rightarrow \CG$ and a 2-homomorphism $\alpha:g\circ f\rightarrow 0$ such that $\CE$ is a principal $\CA$-bundle over $\CG$.
\end{definition}

We are interested in extensions of a smooth 2-group $\CG=(\sG\rightrightarrows \sG)$ by a smooth abelian 2-group $\CA=\sA\rightrightarrows *$, which form the weak 2-category $\sExt(\CG,\CA)$. The following theorem gives a way of encoding this weak 2-category in Segal--Mitchison cohomology classes:
\begin{theorem}[{\cite[Thm.\ 1]{Schommer-Pries:0911.2483}}]\label{extension_2-groups}
 Let $\sG$ be a Lie group and $\sA$ be an abelian Lie group, viewed as a trivial $\sG$-module. Then there is an (unnatural) equivalence of weak symmetric monoidal 2-categories\footnote{Here, the weak 2-categories $M$, $M[1]$ and $M[2]$ are the obvious trivial weak 2-categories with objects $M$, $*$, $*$, morphisms $M$, $M$, $*$ and 2-morphisms $M$, $M$, $M$, respectively.}:
 \begin{equation}
  \sExt(\CG,\sB\sA)\cong H^3_{\rm SM}(\sG,\sA)\times H^2_{\rm SM}(\sG,\sA)[1]\times H^1_{\rm SM}(\sG,\sA)[2]~.
 \end{equation}
\end{theorem}

For the model $\CS_\lambda$ of the string group of $\sSO(n)$, we are interested in the case $\sG=\sSpin(n)$ and $\sA=\sU(1)$. At least for $n\geq 5$, the cohomology groups $H^2_{\rm SM}(\sG,\sA)$ and $H^1_{\rm SM}(\sG,\sA)$ are trivial. Thus, the corresponding extension is parameterized by an element $\lambda=(\lambda^{3,0},\lambda^{2,1},\lambda^{1,2},\lambda^{0,3})$ of $H^3_{\rm SM}(\sG,\sA)$, cf.\ equations~\eqref{eq:cocycle-example}. We now have the following theorem.
\begin{theorem}[{\cite[Thm.\ 100]{Schommer-Pries:0911.2483}}]
 For $n\geq 5$, $H^3_{\rm SM}(\sSpin(n),\sU(1))\cong \RZ$ and the central extension of smooth 2-groups $\CS_\lambda$ corresponding to a generator $\lambda$ gives a smooth 2-group model for $\sString(n)$.
\end{theorem}

Let us now work through the details of this string group model $\CS_\lambda$. Given a simplicial cover $V_\bullet$ of $\sSpin(n)$ as constructed in section~\ref{ssec:diff-hypercohomology}, the $3$-cocycle $\lambda$ contains the non-trivial smooth maps
\begin{equation}
 \lambda^{0,3}:V_3^{[1]}\rightarrow \sA~,~~~\lambda^{1,2}:V_2^{[2]}\rightarrow \sA\eand \lambda^{2,1}:V_1^{[3]}\rightarrow \sA~.
\end{equation}
As remarked in section ~\ref{ssec:diff-hypercohomology}, the map $\lambda^{2,1}$ is in fact a \v Cech 2-cocycle and  defines an $\sA$-bundle gerbe over $\sSpin(n)$. Identifying bundle gerbes with central groupoid extensions, we obtain the groupoid underlying the smooth 2-group corresponding to $\lambda$:
\begin{equation}
\CS_{\lambda}~:=~ V_1^{[2]}\times \sA\rightrightarrows V_1~.
\end{equation}
Here the source, target and identity maps are given by
\begin{equation}\label{eq_source_target_identity_maps}
\sfs(v_0,v_1,a)= v_1~, ~~~ \sft(v_0,v_1,a)= v_0 \eand \sfid(v_0)= (v_0,v_0,0)~,
\end{equation} 
and the invertible composition is defined as 
\begin{equation}\label{vertical_multiplication}
(v_0,v_1,a_0)\circ (v_1,v_2,a_1):=(v_0,v_2,a_0+a_1+\lambda^{2,1}(v_0,v_1,v_2))
\end{equation}
for $v_{0,1,2}\in V_1$ and $a_0, a_1\in \sA$.
It remains to specify the 2-group structure on the Lie groupoid $\CS_{\lambda}$. 

Note that there is a Lie groupoid functor $(\sff_0,\sff_2)$ from the Lie groupoid $\CC_2:=(V_2^{[2]}\times \sA^{\times 2}\rightrightarrows V_2)$ to $\CS_\lambda\times \CS_\lambda$. This functor is a weak equivalence in $\CatManCat$ and upon bundlization, we can invert it.  The same is true for the functor $(\sff_0\sff_0,\sff_2\sff_0,\sff_2\sff_2)$ from the Lie groupoid $\CC_3:=(V_3^{[2]}\times \sA^{\times 3}\rightrightarrows V_3)$ to $\CS_\lambda^{\times 3}$. This yields bibundles 
\begin{equation}\label{eq:equivalences_inv}
B_2:\CS_\lambda\times \CS_\lambda \rightarrow \CC_2\eand B_3:\CS_\lambda\times \CS_\lambda\times \CS_\lambda \rightarrow \CC_3~.
\end{equation}

Furthermore, we have the Lie groupoid functors
\begin{equation}
 \CC_2\stackrel{\sfm}{\longrightarrow}\CS_\lambda~,~~~\CC_3\stackrel{p_1}{\longrightarrow}\CS_\lambda\eand \CC_3\stackrel{p_2}{\longrightarrow}\CS_\lambda~,
\end{equation}
where
\begin{equation}
\begin{aligned}
 \sfm(y_0,y_1,a_0,a_1)&:=(\sff_1(y_0),\sff_1(y_1),a_0+a_1+\lambda^{1,2}(y_0,y_1))~,\\
 p_1(z_0,z_1,a_0,a_1,a_2)&:=(\sff_1\sff_1(z_0),\sff_1\sff_1(z_1),a_0+a_1+a_2+\sff_1^*\lambda^{1,2}(z_0,z_1)+\sff_3^*\lambda^{1,2}(z_0,z_1))~,\\
 p_2(z_0,z_1,a_0,a_1,a_2)&:=(\sff_1\sff_2(z_0),\sff_1\sff_2(z_1),a_0+a_1+a_2+\sff_0^*\lambda^{1,2}(z_0,z_1)+\sff_2^*\lambda^{1,2}(z_0,z_1))
\end{aligned}
\end{equation}
for $y_{0,1}\in V_2^{[2]}$, $z_{0,1}\in V_3^{[2]}$ and $a_{0,1,2}\in \sA$. There is a natural isomorphism $T: V_{3}\longrightarrow V_{1}^{[2]} \times \sU(1)$ defined by 
\begin{align} \label{associator_naturaltrnas}
T(z_0) \ &= \ (\sff_1\sff_1(z_0), \sff_1\sff_2(z_0), \lambda^{0,3}(z_0))~, ~~~z_0 \in V_3~~,
\end{align}
which satisfies
\begin{align}\label{naturality_associator}
  p_2(z_0,z_1,a_0,a_1,a_2)\circ T(z_1)=T(z_0)\circ p_1(z_0,z_1,a_0,a_1,a_2)
\end{align}
due to $\delta_{\check C}\lambda^{0,3}=\delta_N \lambda^{1,2}$.

After bundlization and composition with the bibundles~\eqref{eq:equivalences_inv}, we obtain bibundles 
\begin{equation} \label{eq:horizpntal_mult}
 \begin{aligned}
  B_\sfm:&~~\CS_\lambda\times \CS_\lambda\rightarrow \CS_\lambda~,\\
  B_{p_1}:&~~\CS_\lambda\times \CS_\lambda\times \CS_\lambda \rightarrow \CS_\lambda~,\\
  B_{p_2}:&~~\CS_\lambda\times \CS_\lambda\times \CS_\lambda \rightarrow \CS_\lambda~,
 \end{aligned}
\end{equation}
and the natural isomorphism $T$ yields a bibundle isomorphism $\sfa: B_{p_1}\Rightarrow B_{p_2}$. Because the bibundles $B_{p_1}$ and $B_{p_2}$ can be identified with $B_\sfm\otimes(B_\sfm\times 1)$ and $B_\sfm\otimes(1\times B_\sfm)$, respectively, $\sfa$ is indeed the associator. Here, $\sfa$ is completely determined by $T$ since $\sfa$ is the horizontal composition of $T$ with the identity isomorphism on $B_3$.

It remains to define the unit $\sfe$ as well as the left- and right-unitors $\sfl$ and $\sfr$. Both unitors are trivial (i.e.\ the identity isomorphism) and up to isomorphism, the unit is uniquely defined as the bundlization $\sfe$ of the Lie groupoid functor
\begin{equation}
  (*\rightrightarrows*) \longrightarrow \CS_{\lambda}~,
\end{equation}
which takes $*$ to a $v_0\in V^{[1]}_{1,p}$ with $\pi(v_0)=\unit_\sG$.

Let us now briefly verify that we indeed constructed a smooth 2-group. For this, we need to check that the bibundle $(B_{p_1},B_\sfm)$ is an equivalence and that the internal pentagon identity is satisfied. The former is relatively clear, because $B_2$ and thus also $(B_{p_1},B_2)$ are bibundle equivalences. One then readily checks that
\begin{equation}
(\id\times\hat \sfm):~\CS_\lambda\times\CC_2~\rightarrow~\CS_\lambda\times \CS_\lambda
\end{equation}
is a bibundle equivalence. It is obvious that the associator only affects the $\sA$-part of the Lie groupoids $\CS_\lambda$, $\CC_2$ and $\CC_3$, and therefore the internal pentagon identity reduces to the equation
\begin{equation}\label{eq:lambda-03}
\begin{aligned}
 &\lambda^{0,3}(v_1,v_2,v_3)+\lambda^{0,3}(v_0,v_1v_2,v_3)+\lambda^{0,3}(v_0,v_1,v_2)\\
 &\hspace{3cm}=\lambda^{0,3}(v_0v_1,v_2,v_3)+\lambda^{0,3}(v_0,v_1,v_2v_3)~,
\end{aligned}
\end{equation}
where $v_{0,1,2,3}\in V_1^{[1]}$. This is precisely the equation $\delta_N \lambda^{0,3}=0$, which holds since $\lambda$ is a Segal--Mitchison 3-cocycle. Finally, note that the interchange law, which is the compatibility condition for the vertical and horizontal multiplications, follows from $\delta_N  \lambda^{2,1} = \delta_{\vC}\lambda^{1,2}$. 

We conclude this section with the following two remarks:
\begin{remark}
 While we are mostly interested in the smooth 2-group model of the string group $\CS_\lambda$ given by the central extension of the smooth 2-group $\sSpin(n)\rightrightarrows \sSpin(n)$ by $\sA\rightrightarrows *$, the above construction of this extension as well as most of our following discussion readily generalizes to arbitrary Lie groups $\sG$.
\end{remark}
\begin{remark}
 Multiplicative bundle gerbes as defined in~\cite{Carey:0410013} are special cases of the above construction of a 2-group object internal to $\CatBibun$ from a Segal--Mitchison 3-cocycle. 
\end{remark}

\subsection{From the smooth 2-group model to a weak 2-group model}\label{ssec:weak_2_group_model}

Recall that it has been shown in~\cite{Zhu:0609420} that smooth 2-groups are equivalent  Lie 2-quasigroup\-oids with a single object, which are given by certain Kan simplicial manifolds. The difference to a weak Lie 2-group, which is a weak 2-group object internal to $\CatManCat$, is that in the latter case, horizontal composition of objects and morphisms yields unique objects, which is not true in the case of Lie 2-quasigroupoids. 

In particular, consider horizontal composition of two objects $(v_0,v_1)$ by the composition bibundle $B_\sfm$ in the smooth string 2-group model. The result is a set of isomorphic objects given by $\{\tau(b)|b\in B_\sfm:\sigma(b)=(v_0,v_1)\}$. If the simplicial cover $V_\bullet$ used in the construction of the string group model consists of contractible patches $V_1$, then the bibundle $B_\sfm$, and in particular the bibundle $\CC_2$ is trivial over $V_1\times V_1$ and allows for a global section. By proposition~\ref{prop:bundlization}, $B_\sfm$ is then isomorphic to a bundlization. 

To give the underlying multiplication functor explicitly, we proceed as follows. Without restriction in the cases we are interested in, we assume a simplicial cover $V_\bullet$ as constructed in example~\ref{ex:cover}. In particular, the simplicial index set $I_\bullet$ has now a total order with each subset of the simplicial set having a lowest element. We can now use these lowest elements to fix ambiguities, like defining preferred horn fillers and fixing a unique identity object in $\CS_\lambda$.

First, consider the surjective submersion $(\sff_2,\sff_0):V_2~\rightarrow~V_1\times V_1$. For each element $(v_0,v_1)\in V_1\times V_1$, we can now choose the element of $V_2$ over $(v_0,v_1)$ with the lowest position according to the obvious lexicographic ordering of patch multiindices. This defines a function $\phi^{[1]}_2:(V_1\times V_1)\rightarrow V_2$ satisfying
\begin{equation}
 \sff^2_0 \phi^{[1]}_2(v_0,v_1)=v_1\eand \sff^2_2\phi^{[1]}_2(v_0,v_1)=v_0~.
\end{equation}
In the language of quasigroupoids and Kan complexes, the function $\phi_2$ picks a horn filler in $V_2$ for the horn $(v_0,v_1)\in  V_1\times V_1$. Applying the face map $\sff^2_1$ to this horn filler then yields a preferred horizontal composition:
\begin{equation}
 v_0\otimes v_1:=\sff^2_1\phi^{[1]}_2(v_0,v_1)~.
\end{equation}
Since the lexicographic ordering on $V_2$ arises from that on $V_1$, we evidently have a relation between $\phi^{[1]}_2:V_1\times V_1\rightarrow V_2$ and $\phi_1:\sG\rightarrow V_1$:
\begin{proposition}
  The horizontal composition is completely induced from the product on $\sG$:
  \begin{equation}
    v_0\otimes v_1:=\sff^2_1\phi^{[1]}_2(v_0,v_1)=\phi_1(\pi(v_0)\pi(v_1))
  \end{equation}
  for all $v_{0,1}\in V_1$.
\end{proposition}
\begin{corollary}\label{cor:product_simplify}
 We have the following identities:
 \begin{equation}
 \begin{aligned}
  &\pi(v_0\otimes v_1)=\pi(v_0)\pi(v_1)~,\\
  &v_0\otimes v_1\cong v_2\otimes v_3~~~\Rightarrow~~~v_0\otimes v_1=v_2\otimes v_3~,\\
  &(v_0\otimes v_1)\otimes v_2=v_1\otimes (v_1\otimes v_2)
 \end{aligned}
 \end{equation}
 for all $v_{0,1,2}\in V_1$.
\end{corollary}
\begin{proof}
 The first relation follows from the fact that $\phi_1$ is a section of $\pi$ and therefore $\pi\circ \phi_1=\id_\sG$. The second and third relations are then direct consequences of the first one.
\end{proof}
\noindent Note that the above corollary does not imply that the associator is trivial; it merely has the same source and target. We can now readily extend $\phi_2^{[1]}:V_1\times V_1\rightarrow V_2$ to higher fibered products as done in the following lemma:
\begin{lemma}
 The map $\phi_2^{[2]}:V_1^{[2]}\times V_1^{[2]}\rightarrow V_2^{[2]}$ with
\begin{equation}
 \phi_2^{[2]}\big((v_0,v_1),(v_2,v_3)\big):=\big(\phi^{[1]}_2(v_0,v_2),\phi^{[1]}_2(v_1,v_3))\big)~,
\end{equation}
where $v_{0,1,2,3}\in V_1$ with $\pi(v_0)=\pi(v_1)$ and $\pi(v_2)=\pi(v_3)$, renders the following diagram commutative:
\begin{equation}
     \xymatrixcolsep{6pc}
     \xymatrixrowsep{2pc}
    \myxymatrix{
      \sG\times \sG & V_2\ar@{->}[l]^{\pi_2} & V_2^{[2]} \ar@<-.5ex>[l]\ar@<.5ex>[l] \\
      \sG\times \sG \ar@{->}[u]^{=} & V_1\times V_1\ar@{->}[l]^{\pi_1\times \pi_1} \ar@{->}[u]^{\phi^{[1]}_2}& V_1^{[2]}\times V_1^{[2]} \ar@<-.5ex>[l]\ar@<.5ex>[l]\ar@{->}[u]^{\phi_2^{[2]}}
    }
\end{equation}
\end{lemma}
\noindent The maps $\phi^{[1]}_2$ and $\phi_2^{[2]}$ define a functor internal to $\CatMan$,
\begin{equation}
     \xymatrixcolsep{6pc}
     \xymatrixrowsep{2pc}
    \myxymatrix{
     V_2^{[2]}\times \sA\times \sA \ar@<-.5ex>[d] \ar@<.5ex>[d] &
V_1^{[2]}\times V_1^{[2]}\times \sA\times \sA \ar@<-.5ex>[d]\ar@<.5ex>[d]\ar@{->}[l]^-{\phi^{[2]}_2\times\id_{\sA\times \sA}}\\
    V_2 & V_1\times V_1 \ar@{->}[l]^{\phi^{[1]}_2}
    }
\end{equation}
and we can replace the bibundle $\CC_2$ with the bundlization of this functor, making horizontal composition unique also for morphisms.

We also have a surjective submersion $(\sff^3_0,\sff^3_2,\sff^3_3):V_3\rightarrow V_2\times V_2\times V_2$, and we define a map $\phi^{[1]}_3:V_1\times V_1\times V_1$ as the horn filler of $\phi_2(v_1,v_2)$, $\phi_2(v_0,v_1\otimes v_2)$ and $\phi_2(v_0,v_1)$ with the lowest lexicographic position. It satisfies
\begin{equation}
\begin{aligned}
 \sff^3_0 \phi_3(v_0,v_1,v_2)&=\phi_2(v_1,v_2)~,~~~&\sff^2_0\sff^3_0 \phi_3(v_0,v_1,v_2)&=v_2~,\\
 \sff^3_2\phi_3(v_0,v_1,v_2)&=\phi_2(v_0,v_1\otimes v_2)~,~~~&\sff^2_2\sff^3_0\phi_3(v_0,v_1,v_2)&=v_1~,\\
 \sff^3_3\phi_3(v_0,v_1,v_2)&=\phi_2(v_1,v_2)~,~~~&\sff^2_2\sff^3_2\phi_3(v_0,v_1,v_2)&=v_0~.
\end{aligned}
\end{equation}

Altogether, we arrive at the following theorem.
\begin{theorem}
 The Lie groupoid $\CS_\lambda:=V^{[2]}_1\times \sA\rightrightarrows V_1$, together with the identity-assignment
 \begin{subequations}
 \begin{equation}
    I: (*\rightrightarrows *) \rightarrow \CS_\lambda~,~~~ I_0(*):=\unit_{\CS_\lambda}:=\phi_1(\unit_\sG)~,~~~ I_1(*):=\id_{\unit_{\CS_\lambda}}:=(\unit_{\CS_\lambda},\unit_{\CS_\lambda},0)~,
 \end{equation}
 the horizontal composition 
  \begin{equation}
  \begin{aligned}
    v_0\otimes v_1 &:=\sff^2_1\phi_2(v_0,v_1)=\phi_1(\pi(v_0)\pi(v_1))~,\\
    (v_0,v_1,a_0)\otimes (v_2,v_3,a_1)&:= \big(v_0\otimes v_2,v_1\otimes v_3,a_0+a_1+\lambda^{1,2}(\phi_2(v_0,v_2),\phi_2(v_1,v_3))\big)~,
  \end{aligned}
  \end{equation}
  the vertical composition
  \begin{equation}
    (v_0,v_1,a_0)\circ (v_1,v_2,a_1):=(v_0,v_2,a_0+a_1+\lambda^{2,1}(v_0,v_1,v_2))~,
    \end{equation}
  the unitors
  \begin{equation}
   \sfl_v=(v,\unit_{\CS_\lambda}\otimes v,0)=(v,\phi_1(\pi(v)),0)~,~~~\sfr_v=(v,v\otimes \unit_{\CS_\lambda},0)=(v,\phi_1(\pi(v)),0)
  \end{equation}
  and associator
  \begin{equation}
  \begin{aligned}
    \sfa_{v_0,v_1,v_2}&=\big(~\sff_1\sff_2(\phi_3(v_0,v_1,v_2))~,~\sff_1\sff_1(\phi_3(v_0,v_1,v_2))~,~\lambda^{0,3}(\phi_3(v_0,v_1,v_2))~\big)\\
      &=\big(~v_0\otimes v_1\otimes v_2~,~v_0\otimes v_1 \otimes v_2~,~\lambda^{0,3}(\phi_3(v_0,v_1,v_2))~\big)~,
  \end{aligned}
  \end{equation}
 \end{subequations}
  where $v_{0,1,2,3}\in V_1$ and $a_{0,1}\in \sA$, forms a weak Lie 2-group, which we denote by $\CS_\lambda^{\rm w}$.
\end{theorem}
\noindent Note that since the unitors are non-trivial, $\CS_\lambda^{\rm w}$ is {\em not} a semistrict Lie 2-group in the sense of~\cite{Jurco:2014mva}.

This description of the smooth string 2-group model as a weak Lie 2-group will simplify the explicit computations leading to the cocycle description of principal $\CS_\lambda$-bundles with connection later on.

\section{Differentiation of the string 2-group model}\label{sec:differentiation}

\subsection{Strong homotopy Lie algebras and \texorpdfstring{N$Q$}{NQ}-manifolds}

Clearly, any reasonable differentiation prescription for a categorified Lie group should yield a categorified Lie algebra. The most general notion of categorification of a Lie algebra is that of a weak Lie 2-algebra~\cite{Roytenberg:0712.3461}. Here, we can restrict ourselves to so-called semistrict Lie $n$-algebras, which in turn are categorically equivalent to $n$-term $L_\infty$-algebras~\cite{Baez:2003aa}. In the differentiation method we will use later, the latter will appear in their dual form as N$Q$-manifolds.
\begin{definition}
 An \uline{N$Q$-manifold} is an $\NN$-graded manifold endowed with a vector field $Q$ of degree~1 such that $Q^2=0$. We will refer to the vector field $Q$ as the \uline{Chevalley--Eilenberg differential}.
\end{definition}
N$Q$-manifolds are in one-to-one correspondence with $L_\infty$-algebroids\footnote{Some care has to be taken when homogeneous parts of the N$Q$-manifold become infinite dimensional.}. To get strong homotopy Lie algebras, we need the following restriction.
\begin{definition}
 An \uline{$L_\infty$-algebra} is an N$Q$-manifold concentrated in positive degrees. An \uline{$n$-term $L_\infty$-algebra} is an N$Q$-manifold concentrated in degrees $\{1,\dots,n\}$.
\end{definition}
\noindent For $n=1$, this yields the ordinary Chevalley--Eilenberg description of a Lie algebra.

Let us describe 2-term $L_\infty$-algebras, which are categorically equivalent to semistrict Lie 2-algebras, in more detail. They will play the role of a categorified gauge Lie algebra in our later discussion.
\begin{example}
 Let $X$ and $Y$ be complex vector spaces with coordinates $x^\alpha$ and $y^a$. Then $X[1]\oplus Y[2]$ is an N$Q$-manifold, where the notation implies that elements in $X[1]$ and $Y[2]$ come with homogeneous grading~1 and~2, respectively. The vector field $Q$ is necessarily of the form
 \begin{equation}
  Q=-f^{\alpha}_ay^a\der{x^\alpha}-\tfrac12 f_{\alpha\beta}^\gamma x^\alpha x^\beta \der{x^\gamma}-f^a_{\alpha b}x^\alpha y^b\der{y^a}-\tfrac{1}{3!}f^a_{\alpha\beta\gamma}x^\alpha x^\beta x^\gamma\der{y^a}
 \end{equation}
 with some structure constants $f^{\dots}_{\dots}\in \FC$. The latter define graded antisymmetric multilinear brackets on the shifted space $X[0]\oplus Y[1]$. Introducing the grade-carrying bases $(\tau_\alpha)$ and $(t_a)$ on $X[0]$ and $Y[1]$, respectively, we have
 \begin{equation}
 \begin{aligned}
  \mu_1(t_a)&=f^\alpha_a \tau_\alpha~,~~~&\mu_2(\tau_\alpha,\tau_\beta)&=f_{\alpha\beta}^\gamma \tau_\gamma~,\\
  \mu_2(\tau_\alpha,t_a)&=f^b_{a\alpha}t_b~,~~~&\mu_3(\tau_\alpha,\tau_\beta,\tau_\gamma)&=f_{\alpha\beta\gamma}^a t_a~.
 \end{aligned}
 \end{equation}
 Note that the operations $\mu_i$ are of degree $i-2$ and the condition $Q^2=0$ yields the usual higher or homotopy Jacobi relations between the $\mu_i$ defining a 2-term $L_\infty$-algebra, cf.~\cite{Lada:1992wc,Lada:1994mn}.
\end{example}

\subsection{Cocycle description of principal string 2-group bundles}

For the differentiation of the string 2-group model $\CS_{\lambda}$, we need descent data for principal $\CS_{\lambda}$-bundles in terms of \v Cech cocycles and \v Cech coboundaries. Let us develop these in the following. We restrict ourselves to principal $\CS_\lambda$-bundles over ordinary manifolds subordinate to a cover $U$, and consequently, the covering groupoid $\CU=U\rightrightarrows U$ is discrete. Following the discussion in section~\ref{ssec:cocycle_description}, we start from the diagram
\begin{equation}\label{diag:principal_string_2_bundles}
     \xymatrixcolsep{3pc}
     \xymatrixrowsep{3pc}
    \myxymatrix{
     \CS_\lambda  \ar@<-.5ex>[d] \ar@<.5ex>[d] &
    U^{[2]}\leftleftarrows U^{[2]} \ar@<-.5ex>[d]_{\sfs} \ar@<.5ex>[d]^{\sft}
    \ar@/^2pc/[l]^(.1){}="c"^(.5){}="d"
    \ar@/_2pc/[l]^(.1){}="e"^(.5){}="f" \ar@{=>} "d";"f"^{\beta}
    \ar@/^2pc/[l]^{\Phi_1}
    \ar@/_2pc/[l]_{\Psi_1}\\
    {*}\rightrightarrows {*} & U\leftleftarrows U \ar@{->}[l]^{\Phi_0}_{\Psi_0} 
    }
\end{equation}
where $\CS_\lambda\rightrightarrows (*\leftleftarrows *)$ is a category internal to $\CatBibun$ with composition given by the bibundle $B_\sfm$. The information contained in the functors $\Phi$ and $\Psi$ as well as in the natural isomorphism $\beta$, together with the coherence conditions will yield the appropriate generalization of \v Cech cochains, cocycles and coboundaries describing principal $\CS_\lambda$-bundles and their isomorphisms.

We will assume that the cover $U$ is good and in particular, that $U^{[2]}$ is contractible. This implies that the bibundles $\Phi_1$ and $\Psi_1$ are both trivial bundles over $U^{[2]}$ admitting a global smooth section. By proposition~\ref{prop:bundlization}, this implies that $\Phi$ and $\Psi$ are isomorphic to bundlizations $\hat\phi$ and $\hat\psi$ of smooth functors of Lie groupoids $\phi$ and $\psi$. Moreover, because of proposition~\ref{prop:natural_trafos_1_1}, the bibundle map $\beta$ can be given by a smooth natural transformation between $\phi$ and $\psi$. The only bibundle which is not a bundlization here is the multiplication $B_\sfm$, which appears in the coherence diagrams for internal functors~\eqref{diag:int_functor_coherence_a} and for internal natural transformations~\eqref{diag:int_natural_trafo_coherence_a} with $B_c=B_\sfm$. 

An explicit evaluation of the composition of bibundle isomorphisms in~\eqref{diag:int_functor_coherence_a} is rather cumbersome. To simplify our discussion, we therefore choose to switch to the weak Lie 2-group model $\CS_\lambda^{\rm w}$ of the string 2-group model given in section~\ref{ssec:weak_2_group_model}. Our principal 2-bundles will therefore be weak principal 2-bundles in the sense of~\cite{Jurco:2014mva}, which are given by weak 2-functors internal to $\CatManCat$ from the \v Cech 2-groupoid to the delooping of $\CS_\lambda^{\rm w}$. From there, we also recall the following proposition:
\begin{proposition}[\cite{Jurco:2014mva}, Prop.\ 3.15]\label{prop:normalization}
 Every weak principal 2-bundle $\Phi$ is equivalent to its normalization, which is given by a normalized weak 2-functor which maps to the unit in the structure 2-group over overlaps $U_0\cap U_0$ and whose 2-morphisms are the obvious left and right unitors over $U_0\cap U_0\cap U_1$ and $U_0\cap U_1\cap U_1$.
\end{proposition}
\noindent We will give the consequences of this proposition below; for more details, see~\cite{Jurco:2014mva}. With the above simplifications, we arrive at the following theorem.
\begin{theorem}\label{thm:weak_cocycle}
 The functor $\Phi$ defining a (normalized) principal $\CS_\lambda^{\rm w}$-bundle is described by a 1-cochain $(v_{ij})\in \CC^\infty(U^{[2]},V_1)$ together with a 2-cochain $v_{ijk}\in \CC^\infty(U^{[3]},V^{[2]}_1\times \sA)$ such that
 \begin{equation}\label{eq:weak_cocycle}
 \begin{aligned}
  &v_{ijk}=(v_{ik},v_{ij}\otimes v_{jk},a_{ijk})~,~~~v_{ii}=\unit_{\CS_\lambda}~,~~~v_{iij}=\sfl_{v_{ij}}~,~~~v_{ijj}=\sfr_{v_{ij}}~,\\
  &a_{ikl}+a_{ijk}+\lambda^{1,2}(\phi_2(v_{ik}, v_{kl}),\phi_2(v_{ij}\otimes v_{jk},v_{kl}))\\
  &\hspace{1.4cm}=a_{ijl}+a_{jkl}+\lambda^{1,2}(\phi_2(v_{ij},v_{jl}),\phi_2(v_{ij},v_{jk}\otimes v_{kl}))+\lambda^{0,3}(\phi_3(v_{ij},v_{jk},v_{kl}))~,
 \end{aligned}
 \end{equation}
 where $a_{ijk}\in \CC^\infty(U^{[3]},\sA)$. We call the data $(v_{ij},a_{ijk})$ a \uline{degree-2 \v Cech cocycle} over the cover $U$ with values in $\CS^{\rm w}_\lambda$.
\end{theorem}
\begin{proof}
 The first line of equations in~\eqref{eq:weak_cocycle} is readily derived from $(v_{ijk})$ encoding the natural isomorphism 
 \begin{equation}
  \Phi_{2,ijk}:\Phi_{1,ij}\otimes \Phi_{1,jk}\Rightarrow \Phi_{1,ik}~,
 \end{equation}
 cf.\ the second diagram in~\eqref{diag:int_cat_functor}, together with proposition~\ref{prop:normalization}. The coherence axioms of this natural isomorphism read as 
 \begin{equation}
 \begin{aligned}
  v_{ikl}\circ(v_{ijk}\otimes \id_{v_{kl}})&=v_{ijl}\circ(\id_{v_{ij}}\otimes v_{jkl})\circ \sfa_{v_{ij},v_{jk},v_{kl}}~, \\
  v_{ijj}\circ(\id_{v_{ij}}\otimes \id_{\unit_{\CS_\lambda}})=\sfr_{v_{ij}}&\eand v_{iij}\circ(\id_{\unit_{\CS_\lambda}}\otimes \id_{v_{ij}})=\sfl_{v_{ij}}~,
 \end{aligned}
 \end{equation}
 cf.~\eqref{diag:int_functor_coherence_a}. The last two equations are identities, and the first one reduces to
 \begin{equation}\label{eq:5.7}
 \begin{aligned}
  &\big(v_{il},(v_{ij}\otimes v_{jk})\otimes v_{kl},a_{ikl}+a_{ijk}+\lambda^{1,2}(\phi_2(v_{ik}, v_{kl}),\phi_2(v_{ij}\otimes v_{jk},v_{kl}))+\\ 
  &\hspace{.7cm}+\lambda^{2,1}(v_{il},v_{ik}\otimes v_{kl},(v_{ij}\otimes v_{jk})\otimes v_{kl})\big)\\
  &\hspace{1.4cm}=\big(v_{il},(v_{ij}\otimes v_{jk})\otimes v_{kl},a_{ijl}+a_{jkl}+\lambda^{1,2}(\phi_2(v_{ij},v_{jl}),\phi_2(v_{ij},v_{jk}\otimes v_{kl}))+\\
  &\hspace{2.1cm}+\lambda^{2,1}(v_{il},v_{ij}\otimes v_{jl},v_{ij}\otimes (v_{jk}\otimes v_{kl}))+\lambda^{0,3}(\phi_3(v_{ij},v_{jk},v_{kl}))+\\
  &\hspace{2.1cm}+\lambda^{2,1}(v_{il},v_{ij}\otimes(v_{jk}\otimes v_{kl}),(v_{ij}\otimes v_{jk})\otimes v_{kl})\big)~.
 \end{aligned}
 \end{equation}
 The identities of corollary~\ref{cor:product_simplify} together with the fact that we are working with normalized cocycles cause $\lambda^{2,1}$  to drop out of equation~\eqref{eq:5.7}. The remaining non-trivial part of this equation then yields the equation on the 2-cochain $(a_{ijk})$.
\end{proof}

It is now similarly straightforward to describe the natural 2-isomorphism $\beta$ given the coboundary relation between two \v Cech 2-cocycles.
\begin{theorem}\label{thm:weak_coboundary}
 The natural isomorphism $\beta:\Phi\Rightarrow \Psi$ giving an equivalence relation between (normalized) principal $\CS_\lambda^{\rm w}$-bundles $\Phi$ and $\Psi$ described by 2-cocycles $(v_{ij},a_{ijk})$ and $(v'_{ij},a'_{ijk})$ is captured by 0-cochain and 1-cochains,
 \begin{equation}
  (\beta_i)\in\CC^\infty(U,V_1)\eand (\beta_{ij})\in\CC^\infty(U^{[2]},V_1^{[2]}\times \sA)~,
 \end{equation}
 such that
 \begin{equation}\label{eq:weak_coboundary}
  \begin{aligned}
  &\beta_{ij}=(\beta_i\otimes v'_{ij},v_{ij}\otimes \beta_j,\alpha_{ij})~,~~~\beta_{ii}=\sfr_{\beta_i}^{-1}\circ \sfl_{\beta_i}=\id_{\phi_1(\pi(\beta_i))}~,\\
  &\alpha_{ik}+a_{ijk}+\lambda^{1,2}(\phi_2(v_{ik}, \beta_k),\phi_2(v_{ij}\otimes v_{jk},\beta_k))\\ 
  &\hspace{1.4cm}=\alpha_{ij}+a'_{ijk}+\alpha_{jk}+\lambda^{1,2}(\phi_2(\beta_i,v'_{ik}),\phi_2(\beta_i,v'_{ij}\otimes v'_{jk}))\\
  &\hspace{2.2cm}+\lambda^{0,3}(\beta_i,v'_{ij},v'_{jk})-\lambda^{0,3}(v_{ij},\beta_j,v'_{jk})+\lambda^{0,3}(v_{ij},v_{jk},\beta_k)~.
  \end{aligned}
 \end{equation}
 where $\alpha_{ij}\in \CC^\infty(U^{[2]},\sA)$. We call the data $(\beta_i,\alpha_{ij})$ a \uline{degree-2 \v Cech coboundary} over the cover $U$ with values in $\CS^{\rm w}_\lambda$.
\end{theorem}
\begin{proof}
 The first equation in~\eqref{eq:weak_coboundary} is directly obtain from the defining diagram for $\beta_{ij}$, cf.\ the second diagram in~\eqref{diag:int_natural_trafo_def}. The coherence axioms then read as 
 \begin{equation}
 \begin{aligned}
 &\beta_{ik}\circ(v_{ijk}\otimes \id_{\beta_k})=(\id_{\beta_i}\otimes v'_{ijk})\circ \sfa_{\beta_i,v'_{ij},v'_{jk}}\circ (\beta_{ij}\otimes \id_{v'_{jk}})\circ\\
 &\hspace{5cm}\circ\sfa^{-1}_{v_{ij},\beta_j,v'_{jk}}\circ (\id_{v_{ij}}\otimes \beta_{jk})\circ \sfa_{v_{ij},v_{jk},\beta_k}~,\\
 &\hspace{1.5cm}\beta_{ii}\circ(\id_{\unit_{\CS\lambda}}\otimes \id_{\beta_i})=(\id_{\beta_i}\otimes \id_{\unit_{\CS_\lambda}})\circ \sfr^{-1}_{\beta_i}\circ \sfl_{\beta_i}~,
 \end{aligned}
 \end{equation}
 cf.~\eqref{diag:int_natural_trafo_coherence}, with the second condition directly reducing to the identity
 \begin{equation}
  \beta_{ii}=\sfr^{-1}_{\beta_i}\circ \sfl_{\beta_i}=(\phi_1(\pi(\beta_i)),\phi_1(\pi(\beta_i)),\lambda^{2,1}(\phi_1(\pi(\beta_i)),\beta_i,\phi_1(\pi(\beta_i)))=\id_{\phi_1(\pi(\beta_i))}~.
 \end{equation}
 The part in $V_1^{[2]}$ of the first condition also yields an identity. Thus the component in $\sA$ reads as
 \begin{equation*}
  \begin{aligned}
  &\alpha_{ik}+a_{ijk}+\lambda^{1,2}(\phi_2(v_{ik}, \beta_k),\phi_2(v_{ij}\otimes v_{jk},\beta_k))+\lambda^{2,1}(\beta_i\otimes v'_{ik},v_{ik}\otimes \beta_k,(v_{ij}\otimes v_{jk})\otimes \beta_k)\\
  &\hspace{0.3cm}=\alpha_{ij}+a'_{ijk}+\alpha_{jk}+\lambda^{1,2}(\phi_2(\beta_i,v'_{ik}),\phi_2(\beta_i,v'_{ij}\otimes v'_{jk}))+\lambda^{0,3}(\beta_i,v'_{ij},v'_{jk})\\
  &\hspace{0.8cm}+\lambda^{1,2}(\phi_2(\beta_i\otimes v'_{ij},v'_{jk}),\phi_2(v_{ij}\otimes\beta_j,v'_{jk}))-\lambda^{0,3}(v_{ij},\beta_j,v'_{jk})\\
  &\hspace{0.8cm}+\lambda^{1,2}(\phi_2(v_{ij},\beta_j\otimes v'_{jk}),\phi_2(v_{ij},v_{jk}\otimes\beta_k))
  +\lambda^{0,3}(v_{ij},v_{jk},\beta_k)+\\
  &\hspace{0.8cm}+\lambda^{2,1}(\beta_i\otimes v'_{ik},\beta_i\otimes(v'_{ij}\otimes v'_{jk}),(\beta_i\otimes v'_{ij})\otimes v'_{jk})\\
  &\hspace{0.8cm}+\lambda^{2,1}(\beta_i\otimes(v'_{ij}\otimes v'_{jk}),(\beta_i\otimes v'_{ij})\otimes v'_{jk},(v_{ij}\otimes \beta_j)\otimes v'_{jk})\\
  &\hspace{0.8cm}+\lambda^{2,1}((\beta_i\otimes v'_{ij})\otimes v'_{jk},(v_{ij}\otimes \beta_j)\otimes v'_{jk},v_{ij}\otimes(\beta_j\otimes v'_{jk}))\\
  &\hspace{0.8cm}+\lambda^{2,1}((v_{ij}\otimes \beta_j)\otimes v'_{jk},v_{ij}\otimes(\beta_j\otimes v'_{jk}),v_{ij}\otimes(v_{jk}\otimes \beta_k))\\
  &\hspace{0.8cm}+\lambda^{2,1}(v_{ij}\otimes(\beta_j\otimes v'_{jk}),v_{ij}\otimes(v_{jk}\otimes \beta_k),(v_{ij}\otimes v_{jk})\otimes \beta_k)~.
  \end{aligned}
 \end{equation*}
 Just as in the proof of theorem~\ref{thm:weak_cocycle}, terms containing $\lambda^{2,1}$ drop out due to identities from corollary~\ref{cor:product_simplify}. The same is true for the third and fourth term containing $\lambda^{1,2}$. The remaining part is then coboundary condition on the $\alpha_{ij}$. 
\end{proof}

\subsection{Functor from manifolds to descent data}

To differentiate the smooth 2-group model of the string group, we use a method suggested by \v Severa~\cite{Severa:2006aa}. He observed that the Lie algebra $\frg$ of a Lie group $\sG$ can be regarded as the moduli space of functors from the category of manifolds to descent data of principal $\sG$-bundles on the surjective submersion $N\times \FR^{0|1}\rightarrow N$. In particular, such descent data is given in terms of functions $g(\theta_0,\theta_1):N\times \FR^{0|2}\rightarrow \sG$, which satisfy
\begin{equation}
 g(\theta_0,\theta_1)g(\theta_1,\theta_2)=g(\theta_0,\theta_2)~.
\end{equation}
This relation implies that $g(\theta_0,\theta_1)=g(\theta_0,0)(g(\theta_1,0))^{-1}$ and we can expand\footnote{For simplicity, assume that $\sG$ is a matrix group. Otherwise, one has to insert the diffeomorphism between $\sG$ and $T_\unit\sG$ in an infinitesimal neighborhood of $\unit$ and its inverse into all formulas.}
\begin{equation}
 g(\theta_0,0)=\unit_\sG+\omega \theta_0~,
\end{equation}
where $\omega\in \frg[1]$. We thus recover the Lie algebra as a vector space. The moduli space $\frg[1]$ comes with a natural action of $\sHom(\FR^{0|1},\FR^{0|1})$ and one of its generators can be identified with the Chevalley--Eilenberg differential of $\frg$, encoding the Lie bracket. The natural action of this generator on functions $f$ on $\FR^{0|k}$ reads as 
\begin{equation}
 \dd_{\rm K}f(\theta_0,\theta_1,\dots,\theta_{k-1}):=\dder{\eps}f(\theta_0+\eps,\theta_1+\eps,\dots,\theta_{k-1}+\eps)~,
\end{equation}
and its application to
\begin{equation}\label{eq:ref_g_1}
 g(\theta_0,\theta_1)=\unit_\sG+\omega(\theta_0-\theta_1)+\tfrac12[\omega,\omega]\theta_0\theta_1
\end{equation}
induces an action $\dd_{\rm K}\omega=-\tfrac12[\omega,\omega]$, which in turn yields the Chevalley--Eilenberg differential
\begin{equation}
 Q \xi^\alpha=-\tfrac12 f^\alpha_{\beta\gamma} \xi^\beta\wedge  \xi^\gamma~,
\end{equation}
where the $ \xi^\alpha$ are the coordinate functions on $\frg[1]$.

\v Severa pointed out that this construction extends to higher principal bundles with arbitrary Kan simplicial complexes as structure quasi-groupoids, and it yields a differentiation of the latter to $L_\infty$-algebroids. The definition relevant for our purposes is then the following.
\begin{definition}
 The \uline{Lie $n$-algebra of a smooth $n$-group $\CG$} is the moduli space of functors taking a manifold $N$ to descent data of principal $\CG$-bundles with respect to the surjective submersion $N\times \FR^{0|1}\twoheadrightarrow N$. The algebra structure is encoded in the Chevalley--Eilenberg differential given by a generator of the action of $\sHom(\FR^{0|1},\FR^{0|1})$ onto the moduli space.
\end{definition}
\noindent We are now interested in the special case of principal $\CS_\lambda^{\rm w}$-bundles. Our discussion will follow closely that of~\cite{Jurco:2014mva}, with generalized arguments due to $\CS_\lambda^{\rm w}$ having non-trivial unitors.

We start from a weak normalized 2-functor encoded in a degree-2 \v Cech cocycle $v$ on the cover $N\times \FR^{0|1}\twoheadrightarrow N$. This cocycle consists of a $V_1$-valued \v Cech 1-cochain together with a $V_1^{[2]}\times \sU(1)$-valued 2-cochain,
\begin{equation}
v(\theta_0,\theta_1) \eand v(\theta_0,\theta_1,\theta_2)=\big(v(\theta_0,\theta_2),v(\theta_0,\theta_1)\otimes v(\theta_1,\theta_2),a(\theta_0,\theta_1,\theta_2)\big)~.
\end{equation}
Since $v$ is normalized, we have $v(\theta_0,\theta_0)=v(0,0)=\unit_{\CS_\lambda}$. Note that in a cover $V_1$ as constructed in example~\ref{ex:cover}, open sets in $V_1$ are fully contained within one of the patches. Because $v(\theta_0,\theta_1)$ depends smoothly on the Gra\ss mann variables, it lies on the same patch $V_{1,i}$ as $\unit_{\CS_\lambda}$. This patch contains an infinitesimal neighborhood of $\unit_{\CS_\lambda}$, and we have $\phi_1(\pi(v(\theta_0,\theta_1)))=v(\theta_0,\theta_1)$, which leads to significant simplifications. In particular, the equation $\pi(v(\theta_0,\theta_1))\pi(v(\theta_1,\theta_2))=\pi(v(\theta_0,\theta_2))$ implied by the cocycle condition now lifts to $V_1$:
\begin{equation}
 v(\theta_0,\theta_1)\otimes v(\theta_1,\theta_2)=v(\theta_0,\theta_2)~.
\end{equation}
This, in turn, renders the $\lambda^{1,2}$-terms in the cocycle condition~\eqref{eq:weak_cocycle} for the $a(\theta_0,\theta_1,\theta_2)$ trivial, leaving us with
\begin{equation}
\begin{aligned}
   a(\theta_0,\theta_2,\theta_3)+&a(\theta_0,\theta_1,\theta_2)=\\
   &\hspace{0.5cm}a(\theta_0,\theta_1,\theta_3)+a(\theta_1,\theta_2,\theta_3)+\lambda^{0,3}(\phi_3(v(\theta_0,\theta_1),v(\theta_1,\theta_2),v(\theta_2,\theta_3))~.
\end{aligned}
\end{equation}

As one might expect due to the form of the surjective submersion, the principal $\CS_\lambda^{\rm w}$-bundle we are dealing with here is trivial.
\begin{lemma}
 The cochain $\beta$ defined by 
  \begin{equation}
  \beta(\theta_0):=v(\theta_0,0)\eand \alpha(\theta_0,\theta_1):=a(\theta_0,\theta_1,0)
  \end{equation}
 forms a coboundary as defined in theorem~\ref{thm:weak_coboundary}, trivializing the principal $\CS_\lambda^{\rm w}$-bundle described by $v$. Moreover, we have $\alpha(\theta_0,0)=0$ as well as $\alpha(0,\theta_0)=0$.
\end{lemma}
\begin{proof}
 This follows by direct computation\footnote{see also~\cite{Jurco:2014mva} for some of the technical details}, using the fact that the $\lambda^{(p,q)}$ vanish if an argument is given by a sequence of degeneracy maps acting on $\unit_{\CS_\lambda}$. Note in particular that $a(\theta_0,\theta_0,\theta_1)=a(\theta_0,\theta_1,\theta_1)=0$ because of the normalization of $v$.
\end{proof}

We can now fix the following expansion of the cochain $\beta$ in the Gra\ss mann variables, using implicitly the local diffeomorphism between the neighborhood of $\unit_{\CS_\lambda}$ and $T_{\unit_{\CS_\lambda}} V_1$:
\begin{equation}
 \begin{aligned}
  \beta(\theta_0)&=\unit_{\CS_\lambda}+\omega\theta_0\eand \alpha(\theta_0,\theta_1)&=\psi\theta_0\theta_1~,
 \end{aligned}
\end{equation}
where $\omega\in T_{\unit_{\CS_\lambda}}[1]V_1\cong T_{\unit_\sG}[1]\sG=\sLie(\sG)[1]$ and $\psi\in \sLie(\sA)[2]$. 

Next, we use the explicit coboundary relations to compute the Gra\ss mann expansion of the cocycle components. The fact that in a neighborhood of $\unit_{\CS_\lambda}$, the horizontal composition collapses to group multiplication in $\sG$ directly yields
\begin{equation}
  v(\theta_0,\theta_1)=\unit_{\CS_\lambda}+\omega(\theta_0-\theta_1)+\tfrac12[\omega,\omega]\theta_0\theta_1~,
\end{equation}
cf.\ equation~\eqref{eq:ref_g_1}. On the other hand, the coboundary relation for the morphisms reduces to 
\begin{equation}
 \alpha(\theta_0,\theta_2)+a(\theta_0,\theta_1,\theta_2)=\alpha(\theta_0,\theta_1)+\alpha(\theta_1,\theta_2)+\lambda^{0,3}(v(\theta_0,\theta_1),v(\theta_1,\theta_2),\beta(\theta_2))~.
\end{equation}
Since $\lambda^{0,3}(\unit_{\CS_\lambda},-,-)=\lambda^{0,3}(-,\unit_{\CS_\lambda},-)=\lambda^{0,3}(-,-,\unit_{\CS_\lambda})=0$, we have
\begin{equation}
 \lambda^{0,3}(v(\theta_0,\theta_1),v(\theta_1,\theta_2),\beta(\theta_2))=:\lambda^{0,3}(\omega,\omega,\omega)\theta_0\theta_1\theta_2~,
\end{equation}
where $\lambda^{0,3}(\omega,\omega,\omega)$ is the obvious linearization of $\lambda^{0,3}$ around $(\unit_{\CS_\lambda},\unit_{\CS_\lambda},\unit_{\CS_\lambda})$. From this, we compute the expansion
\begin{equation}
 \begin{aligned}
  a(\theta_0,\theta_1,\theta_2)&=\psi(\theta_0\theta_1+\theta_1\theta_2-\theta_0\theta_2)+\lambda^{0,3}(\omega,\omega,\omega)\theta_0\theta_1\theta_2~.
 \end{aligned}
\end{equation}
The Chevalley--Eilenberg differential induced by the relevant generator of $\sHom(\FR^{0|1},\FR^{0|1})$ is then characterized by
\begin{equation}
 \begin{aligned}
  \dd_{\rm K} \omega&=-\tfrac12 [\omega,\omega]~,\\
  \dd_{\rm K} \psi&=-\lambda^{0,3}(\omega,\omega,\omega)~.
 \end{aligned}
\end{equation}
We can summarize our findings in the following theorem.
\begin{theorem}
  The Lie 2-algebra of the smooth 2-group model $\CS^{\rm w}_{\lambda}$ of the string group with $\sG=\sSpin(n)$ is the string Lie 2-algebra equivalent to the 2-term $L_{\infty}$-algebra $\fru(1)[1]\rightarrow \sLie(\sG)$, together with the non-trivial higher products 
  \begin{equation}
   \mu_2(x_1,x_2)=[x_1,x_2]~,~~~\mu_3(x_1,x_2,x_3)=k(x_1,[x_2,x_3])
  \end{equation}
  for some $k\in \FR$. Here $(-,-)$ and $[-,-]$ denote the Killing form and the Lie bracket on $\sLie(\sG)$, respectively.
\end{theorem}
\begin{proof}
 It only remains to argue that $\lambda^{0,3}(\omega,\omega,\omega)\sim k(\omega,[\omega,\omega])$ for some $k\in\FR$. For the Lie groups considered in this theorem, any such 3-cocycle is necessarily of this form.
\end{proof}

\subsection{Equivalence transformations}

Let us now extend the above differentiation process to derive equivalence relations on the moduli space. Later we will exchange the Chevalley--Eilenberg differential for the de Rham differential and this will give us gauge transformations for connections on principal $\CS_\lambda^{\rm w}$-bundles and further the full underlying Deligne cohomology of these bundles. This approach to Deligne cohomology with values in categorified groups was first used in~\cite{Jurco:2014mva}.

Given a principal $\CS_\lambda^{\rm w}$-bundle on $N\times\FR^{0|1}\twoheadrightarrow N$ in terms of a \v Cech 2-cocycle $v$, we now perform an isomorphism $\beta$ to another such \v Cech 2-cocycle $v'$:
\begin{equation}
 \big(\beta(\theta_0),\alpha(\theta_0,\theta_1)\big)~:~\big(v(\theta_0,\theta_1),a(\theta_0,\theta_1,\theta_2)\big)~\rightarrow~\big(v'(\theta_0,\theta_1),a'(\theta_0,\theta_1,\theta_2)\big)~.
\end{equation}
At the level of moduli, this translates into a relation 
\begin{equation}
 \big(\beta(\theta_0),\beta(\theta_0,\theta_1)\big)~:~(\omega,\psi)~\rightarrow~(\omega',\psi')
\end{equation}
and we are interested in the explicit isomorphism. The \v Cech 2-coboundary $\beta$ is necessarily of the form
\begin{equation}
 \beta(\theta_0)=\beta-\dd_K \beta\,\theta_0\eand \alpha(\theta_0,\theta_1)=\zeta(\theta_1-\theta_0)+\dd_K\zeta\,\theta_0\theta_1~.
\end{equation} 
Because both $v$ and $v'$ are normalized, $\beta$ relates the 2-cocycles as follows:
\begin{equation}\label{eq:equivalence_descent_data1}
  \begin{aligned}
    &v(\theta_0,\theta_1)\otimes \beta(\theta_1)=\beta(\theta_0)\otimes v'(\theta_0,\theta_1)~,\\
    &\alpha(\theta_0,\theta_2)+a(\theta_0,\theta_1,\theta_2)=\alpha(\theta_0,\theta_1)+a'(\theta_0,\theta_1,\theta_2)+\alpha(\theta_1,\theta_2)\\
    &\hspace{2cm}+\lambda^{0,3}(\beta(\theta_0),v'(\theta_0,\theta_1),v'(\theta_1,\theta_2))-\lambda^{0,3}(v(\theta_0,\theta_1),\beta(\theta_1),v'(\theta_1,\theta_2))\\
    &\hspace{2cm}+\lambda^{0,3}(v(\theta_0,\theta_1),v(\theta_1,\theta_2),\beta(\theta_2))~,
  \end{aligned}
\end{equation}
and the second equation reduces to
\begin{equation}\label{eq:equivalence_descent_data2}
\begin{aligned}
 &\psi(\theta_0\theta_1+\theta_1\theta_2-\theta_0\theta_2)=\big(\psi'+\dd_K\zeta+\lambda^{0,3}(\beta,\omega',\omega')-\lambda^{0,3}(\omega,\beta,\omega')+\lambda^{0,3}(\omega,\omega,\beta)\big)\times\\&\hspace{6cm}\times (\theta_0\theta_1+\theta_1\theta_2-\theta_0\theta_2)~,\\
 &\lambda^{0,3}(\omega,\omega,\omega)=\lambda^{0,3}(\omega',\omega',\omega')+\lambda^{0,3}(\beta,[\omega,\omega],\omega)+\dots-\lambda^{0,3}(\dd_K\beta,\omega',\omega')-\dots~.
\end{aligned}
\end{equation}
From the first equation in~\eqref{eq:equivalence_descent_data1} and the first equation in~\eqref{eq:equivalence_descent_data2}, we readily read off the following relations:
\begin{equation}\label{eq:relation_moduli}
\begin{aligned}
 \beta\otimes \omega'&=\omega\otimes \beta+\dd_K\beta~,\\
 \psi'&=\psi-\dd_K \zeta-\lambda^{0,3}(\beta,\omega',\omega')+\lambda^{0,3}(\omega,\beta,\omega')-\lambda^{0,3}(\omega,\omega,\beta)~.
\end{aligned}
\end{equation}
The second equation of~\eqref{eq:equivalence_descent_data2} is then automatically satisfied, and we arrive at the following theorem.
\begin{theorem}
 \v Cech 2-coboundaries between \v Cech 2-cocycles $v=(\omega,\psi)$ and $v'=(\omega',\psi')$ corresponding to descent data for principal $\CS_\lambda^{\rm w}$-bundles on surjective submersions of the form $N\times \FR^{0|1}\twoheadrightarrow N$ are parameterized by elements $\beta\in V_1$ and $\zeta\in \sA[1]$. The moduli of the coboundaries and cocycles are related as in equation~\eqref{eq:relation_moduli}.
\end{theorem}

\section{Gauge theory with the string 2-group} \label{Section:gauge_theory_with_string_grp}

\subsection{Local description with infinitesimal gauge symmetries}

The local description of higher gauge theory with the string 2-group is readily given without our above considerations and below we briefly recall how, cf.\ e.g.~\cite{Jurco:2014mva}. The string Lie 2-algebra of a connected compact simple Lie group is known to be $\au(1)\rightarrow \frg$ with non-trivial higher products $\mu_2(x_1,x_2)=[x_1,x_2]$ and $\mu_3(x_1,x_2,x_3)=k(x_1,[x_2,x_3])$, $x_i\in \frg$ and $k\in \FR$. This $L_\infty$-algebra can be tensored with the graded differential algebra of differential forms on a contractible patch $U$ of a manifold, $\Omega^\bullet(U)$. The result is another $L_\infty$-algebra, $\tilde \sL$ with higher products $\tilde \mu_i$. The latter are the tensor products of the higher products $\mu_i$ on the string Lie 2-algebra with the differential on $\Omega^\bullet(U)$.

Recall that in any $L_\infty$-algebra, we can define homotopy Maurer--Cartan elements. In $\tilde \sL$, these are elements $\phi$, for which the {\em homotopy Maurer--Cartan equation}
\begin{equation}\label{eq:MCeqs}
 \sum_{i=1}^\infty \frac{(-1)^{i(i+1)/2}}{i!}\tilde \mu_i(\phi,\dots,\phi) = 0
\end{equation}
is satisfied. This equation exhibits a gauge symmetry, parameterized at infinitesimal level by an element $\gamma\in \tilde L$ of degree 0, which maps Maurer--Cartan elements to Maurer--Cartan elements:
\begin{equation}\label{eq:MCgaugetrafos}
 \phi\rightarrow \phi+\delta \phi\ewith \delta \phi = \sum_i \frac{(-1)^{i(i-1)/2}}{(i-1)!}\tilde \mu_i(\gamma,\phi,\dots,\phi)~.
\end{equation}
More explicitly, we have the following proposition.
\begin{proposition}
 The homotopy Maurer--Cartan elements of $\tilde \sL$ are given by pairs $A\in \Omega^1(U)\otimes \frg$ and $B\in \Omega^2(U)\otimes \au(1)$ satisfying the equations
  \begin{equation}\label{eq:field_equations}
  \begin{aligned}
  \CF&\ :=\ \dd A+\tfrac{1}{2}\mu_2(A,A)\ =\ 0~,\\
  H&\ :=\ \dd B-\tfrac{1}{3!}\mu_3(A,A,A)\ =\ 0~.
  \end{aligned}
  \end{equation}
  Infinitesimal gauge transformations are parameterized by pairs $x\in \Omega^0(U)\otimes \frg$ and $\zeta\in \Omega^1(U)\otimes \au(1)$ and act according to
  \begin{equation}\label{eq:gauge_trafos}
    \delta A=\dd x+\mu_2(A,x)\eand\delta B=-\dd \zeta+\tfrac{1}{2}\mu_3(x,A,A)~.
  \end{equation}
\end{proposition}
\begin{proof}
 Substituting $\phi=A-B$ and $\gamma=x+\zeta$ into~\eqref{eq:MCeqs} and~\eqref{eq:MCgaugetrafos} yields the proposition.
\end{proof}

\subsection{Non-abelian Deligne cohomology with values in the string 2-group}

The full global description of non-abelian gauge theory is governed by non-abelian Deligne cohomology. Let us first review the case of ordinary principal bundles with connection before presenting the details for principal $\CS_\lambda^{\rm w}$-bundles.

Given a Lie group $\sG$ with Lie algebra $\frg=\sLie(\sG)$, a principal $\sG$-bundle with connection over a manifold $M$ with respect to a cover $U=\sqcup_i U_i\twoheadrightarrow M$ is described by a non-abelian Deligne 1-cocycle. Such a 1-cocycle consists of $\sG$-valued transition functions $(g_{ij})$ on the fibered product $U\times_M U=\sqcup_{i,j} U_i\cap U_j$ and $\frg$-valued one-forms $(A_i)$ on $U$ satisfying 
\begin{equation}
 g_{ij}g_{jk}=g_{ik}\eand A_j=g_{ij}^{-1} A_ig_{ij}+g_{ij}^{-1}\dd g_{ij}~.
\end{equation}
Note that the cocycle conditions glue together the local data contained in $(A_i)$ to a global connection. Two such Deligne 1-cocycles $(g_{ij},A_i)$ and $(g'_{ij},A'_i)$ are considered equivalent, if they are related by a Deligne 1-coboundary consisting of $\sG$-valued functions on $U$, $(g_i)$, as follows:
\begin{equation}
 g'_{ij}=g_i^{-1}g_{ij}g_j\eand A'_i=g_i^{-1}A_i g_i+g_i^{-1}\dd g_i~.
\end{equation}
The coboundary relations describe finite gauge transformations of the 1-cocycles.

From our discussion in section~\ref{sec:differentiation}, we can now derive the explicit form of Deligne 2-cocycles and 2-coboundaries for principal $\CS_\lambda^{\rm w}$-bundles. In the following, $M$ denotes the base manifold of the principal $\CS_\lambda^{\rm w}$-bundle and $U=\sqcup_i U_i\twoheadrightarrow M$ is a cover of $M$.
\begin{definition}
 Let $\sG$ be a connected compact simple Lie group with Lie algebra $\frg:=\sLie(\sG)$. Let furthermore $\sLie(\CS_\lambda^{\rm w})=(\frg\times \au(1)\rightrightarrows \frg)$ be the Lie 2-algebra of the weak 2-group model $\CS_\lambda^{\rm w}$ over $\sG$. A \uline{Deligne 2-cocycle with values in $\CS_\lambda^{\rm w}$}, which describes a principal $\CS_\lambda^{\rm w}$-bundle with connective structure, is then given by a $\CS_\lambda^{\rm w}$-valued \v Cech 2-cocycle $(v_{ij},a_{ijk})$ together with local forms $A=(A_i)\in \Omega^1(U,\frg)$, $B=(B_i)\in \Omega^2(U,\au(1))$ and forms on overlaps $\zeta=(\zeta_{ij})\in \Omega^1(U\times_M U,\au(1))$, satisfying the following cocycle relations:
 \begin{equation}
  \begin{aligned}
  &\pi(v_{ik})=\pi(v_{ij}\otimes v_{jk})~,~~~v_{ii}=\unit_{\CS_\lambda}~,\\
  &a_{ikl}+a_{ijk}+\lambda^{1,2}(\phi_2(v_{ik}, v_{kl}),\phi_2(v_{ij}\otimes v_{jk},v_{kl}))\\
  &\hspace{1.4cm}=a_{ijl}+a_{jkl}+\lambda^{1,2}(\phi_2(v_{ij},v_{jl}),\phi_2(v_{ij},v_{jk}\otimes v_{kl}))+\lambda^{0,3}(\phi_3(v_{ij},v_{jk},v_{kl}))~,\\
  &\pi(v_{ij})A_i=A_j\pi(v_{ij})+\dd \pi(v_{ij})~,~~~\dd A_i+\tfrac12[A_i,A_i]=0~,\\
  &B_i=B_j-\dd \zeta_{ij}-\lambda^{0,3}(v_{ij},A_i,A_i)+\lambda^{0,3}(A_j,v_{ij},A_i)-\lambda^{0,3}(A_j,A_j,v_{ij})~,\\
  &\zeta_{kj}+\lambda^{0,3}(A_j,v_{ji},v_{ij})=\zeta_{ij}+\zeta_{ki}+\dd a_{jik}+\lambda^{0,3}(v_{ji},A_k,v_{ik})-\lambda^{0,3}(v_{ji},v_{ik},A_k)~.
  \end{aligned}
 \end{equation}
\end{definition}
\noindent Here, the transformations of $A$ and $B$ on double overlaps of patches are the previously derived gauge transformations. The transformations of the $\zeta_{ij}$ are compatibility conditions for the previous transformations on triple overlaps.
\begin{definition}
 A \uline{Deligne 2-coboundary} between two Deligne 2-cocycles $(v,a,A,B,\zeta)$ and $(v',a',A',B',\zeta')$ is a set of functions $\beta=(\beta_i)\in\CC^\infty(U,V_1)$, local 1-forms $\zeta=(\zeta_{i})\in \Omega^1(U,\au(1))$ and functions on overlaps $\alpha=(\alpha_{ij})\in \CC^\infty(U\times_M U,\sU(1))$ such that
 \begin{equation}
  \begin{aligned}
    &\beta_i\otimes v'_{ij}=v_{ij}\otimes \beta_j~,~~\beta_{ii}=\id_{\phi_1(\pi(\beta_i))}~,\\
    &\alpha_{ik}+a_{ijk}+\lambda^{1,2}(\phi_2(v_{ik}, \beta_k),\phi_2(v_{ij}\otimes v_{jk},\beta_k))\\ 
    &\hspace{1.4cm}=\alpha_{ij}+a'_{ijk}+\alpha_{jk}+\lambda^{1,2}(\phi_2(\beta_i,v'_{ik}),\phi_2(\beta_i,v'_{ij}\otimes v'_{jk}))\\
    &\hspace{2.2cm}+\lambda^{0,3}(\beta_i,v'_{ij},v'_{jk})-\lambda^{0,3}(v_{ij},\beta_j,v'_{jk})+\lambda^{0,3}(v_{ij},v_{jk},\beta_k)~,\\
    &\pi(\beta_{i})A'_i=A_i\pi(\beta_{i})+\dd \pi(\beta_{i})~,\\
    &B'_i=B_i-\dd \zeta_{i}-\lambda^{0,3}(\beta_{i},A'_i,A'_i)+\lambda^{0,3}(A_i,\beta_{i},A'_i)-\lambda^{0,3}(A_i,A_i,\beta_{i})~,\\
    &\zeta_{ji}-\lambda^{0,3}(A_i,v_{ij},\beta_j)+\lambda^{0,3}(A_i,\beta_i,v'_{ij})+\lambda^{0,3}(\beta_i,v'_{ij},A'_j)+\zeta_j\\
    &\hspace{1.4cm}=\zeta'_{ji}-\lambda^{0,3}(v_{ij},A_j,\beta_j)+\lambda^{0,3}(v'_{ij},\beta_j,A'_j)+\dd \alpha_{ij}+\lambda^{0,3}(\beta_i,A'_i,v'_{ij})+\zeta_i~.\\
  \end{aligned}
 \end{equation}
 Two Deligne 2-cocycles related by a Deligne 2-coboundary are called \uline{equivalent}.
\end{definition}
\noindent As a consistency check, consider the Deligne 2-cocycle relations and remove one index, say $k$. Then relabel all affected cochains as their corresponding parts of a 2-coboundary, e.g.\ $v_{ik}\rightarrow v_{i}\rightarrow \beta_i$. The resulting relations have to agree with the relations for a 2-coboundary between a Deligne 2-cocycle and the trivial Deligne 2-cocycle, which they do.

The above two definitions provide all necessary details for a global description of the kinematical part of higher gauge theory with the string 2-group $\CS_\lambda^{\rm w}$ as structure group. 

\subsection{Application: A self-dual string solution}

As stated in the introduction, one of the most pressing issues in higher gauge theory is the explicit construction of physically well-motivated examples of higher principal bundles with connective structure. Obvious candidates for dynamical constraints on such connections are the self-duality equation for a 3-form curvature on $\FR^{1,5}$, as well as the self-dual string equation in $\FR^4$. The former should be closely related to a non-abelian formulation of the long-sought (2,0) superconformal field theory in six-dimensions; the latter will be considered in some detail below in a simple case which is readily discussed.

Recall that $k$ abelian self-dual strings~\cite{Howe:1997ue} are described by a 2-form $B$ together with a function $\Phi$ satisfying the equation $H=\dd B=*\dd \Phi$ on $M=\FR^4\backslash\{x_1,\dots,x_k\}$. The points $x_1,\dots,x_k$ are identified with the locations of the $k$ self-dual strings, and the Higgs field $\Phi$ as well as the 2-form potential $B$ diverge at these points. 

One can readily translate the self-dual string equation to higher gauge theory, cf.\ e.g.~\cite{Saemann:2012uq} or~\cite{Palmer:2013haa}. Note that the spin group of the underlying spacetime was intrinsically linked to the gauge group in the 't Hooft--Polyakov monopole~\cite{Prasad:1975kr} and the BPST instanton~\cite{Belavin:1975:85,Atiyah:1979iu}. Therefore, a good choice of an interesting gauge 2-group for self-dual strings is the string group $\CS_\lambda^{\rm w}$ of $\sSpin(4)\cong \sSU(2)\times \sSU(2)$. 

The self-dual string involving the connection of a principal $\CS_\lambda^{\rm w}$-bundle is described by a $\aspin(4)$-valued 1-form $A$ on $M$ together with a $\au(1)$-valued 2-form $B$ and a $\au(1)$-valued function $\Phi$. These have to satisfy the equations
\begin{equation}\label{eq:non_abelian_sds}
 \CF=\dd A+\tfrac12[A,A]=0\eand H=\dd B-\tfrac{1}{3!}\mu_3(A,A,A)=*\dd \Phi~,
\end{equation}
where the first equation is the {\em fake-curvature condition} and the second equation is the non-abelian version of the self-dual string equation. The fake-curvature condition implies here that $A$ is pure gauge, and we write $A=\pi(v)^{-1}\dd \pi(v)$ for some $v\in \CC^\infty(M,V_1)$. Since $\dd \mu_3(A,A,A)=0$, we then have $\dd*\dd \Phi=0$ and $\Phi$ is a harmonic function on $M$. To choose a simple example, we put $M=\FR^4\backslash\{0\}$ and $\Phi=\frac{1}{r^2}$, where $r$ is the distance from the origin in $\FR^4$. There are now two extreme solutions which satisfy the second equation in~\eqref{eq:non_abelian_sds}. One is $v=\unit_{\CS_\lambda^{\rm w}}$ and 
\begin{equation}
 B=\tfrac{3}{8}\dd x^\mu\wedge \dd x^\nu\eps_{\mu\nu\kappa\lambda}\frac{x^\kappa\left(R^2 \arctan\left(\frac{r^\lambda}{x^\lambda}\right)-r^\lambda x^\lambda\right)}{R^2 (r^\lambda)^3}~,~~~R=|x|~,~~~r^\lambda=\sqrt{|x|^2-(x^\lambda)^2}~,
\end{equation}
the other being $B=0$ and 
\begin{equation}
 v=\phi_1\left(\frac{1}{|x|}\left(\begin{array}{cc}
	    x^1+\di x^2 & x^3+\di x^4\\
	    -x^3+\di x^4 & x^1-\di x^2          
         \end{array}\right)~,~\unit~\right)~.
\end{equation}
The first one is a reformulation of the solution given in~\cite{Nepomechie:1984wu}, the second one is an adaptation of the standard gerbe over $S^3\sim \sSU(2)$. Note that the content of $v$ in the second solution can be partially gauged into the other $\sSU(2)$ contained in $\sSpin(4)$. 

We can now show that the above two solutions are indeed gauge equivalent, as one would expect. First, recall that the self-dual string equation on $\FR^4$ can be augmented to a self-duality equation on $\FR^{1,5}$ by assuming that all fields and forms are constant in the two additional direction and identifying $\Phi$ with the 2-form potential in these directions. The gauge transformations of $\Phi$ should therefore be identified with those of this component of $B$, which vanishes if all gauge parameters are constant and the 1-form potentials vanish in these directions. It follows that $\Phi$ is gauge invariant and the same holds trivially for $H$. Since both solutions have the same $\Phi$ and thus the same $H$, this implies that they are gauge equivalent.

Altogether, we found that a solution to the non-abelian self-dual string equations~\eqref{eq:non_abelian_sds} is gauge equivalent to the usual abelian solution. The reason for this was the fake curvature relation $\CF=0$. This equation guarantees that parallel transport along surfaces is reparameterization invariant, see e.g.~\cite{Schreiber:2008aa}. Even though this relation appears naturally in a twistor construction of non-abelian self-dual strings~\cite{Saemann:2012uq}, this equation is not physically relevant for self-dual strings, as the string is perpendicular to the space $M$ and there is no parallel transport within $M$. For more details on this point, see~\cite{Palmer:2013haa}. A study of non-abelian self-dual string solutions which do not satisfy the fake curvature relation is beyond our scope here, and we postpone it to future work.
 
\section*{Acknowledgements}

CS would like to thank Chenchang Zhu for a helpful discussion. This work was partially completed during the workshop ``Higher Structures in String Theory and Quantum Field Theory'' at the Erwin Schr\"odinger International Institute for Mathematical Physics and CS would like to thank the organizers and the institute for hospitality. The work of GAD is supported by a MACS Global Platform Studentship of Heriot--Watt University. The work of CS was partially supported by the Consolidated Grant ST/L000334/1 from the UK Science and Technology Facilities Council.

\appendices

\subsection{Group objects in categories}\label{app:A}

Recall that a category with finite products has a terminal object $*$ and products between any two objects.
\begin{definition}
A \uline{group object} in a category $\CCC$ with finite products is an object $G\in \CCC$ together with morphisms $\sfm\in \CCC(G\times G,G)$, $\sfe\in \CCC(*,G)$, $\inv\in \CCC(G,G)$ such that the following diagrams are commutative:
\begin{equation}
\begin{aligned}
  &\xymatrixcolsep{4pc}
    \myxymatrix{
    G\times G\times G\ar@{->}[r]^(.55){\id_G\times \sfm} \ar@{->}[d]_{\sfm\times \id_G} & G\times G \ar@{->}[d]^{\sfm}\\
    G\times G \ar@{->}[r]^{\sfm} & G
    }
  \hspace{0.5cm}  
  \xymatrixcolsep{4pc}
    \myxymatrix{
    {*}\times G=G\times {*}\ar@{->}[r]^{\id_G\times \sfe} \ar@{->}[d]_{\sfe\times \id_G} \ar@{->}[dr]^{=}& G\times G \ar@{->}[d]^{\sfm}\\
    G\times G \ar@{->}[r]^{\sfm} & G
    }\\
  &\hspace{2cm}\xymatrixcolsep{5pc}
    \myxymatrix{
    G\ar@{->}[r]^{\Delta} \ar@{->}[ddr] & G\times G \ar@{->}[r]^{\inv\times \id_G} \ar@{->}[d]^{\id_G\times \inv} & G\times G\ar@{->}[dd]^{\sfm}\\
    & G\times G \ar@{->}[dr]^{\sfm} &\\
    & {*} \ar@{->}[r]^{\sfe} & G
    }
\end{aligned}
\end{equation}
\end{definition}
\noindent where $\Delta(g)=(g,g)$. A group is then a group object in $\CatSet$ and a Lie group is a group object in $\CatMan$.

\begin{definition}
 Given a group object $G$ in a category with finite products $\CCC$, a $G$-object in $\CCC$ is an object $X\in \CCC$ together with a morphism $\alpha:G\times X\rightarrow X$ such that the following diagrams commute:
\begin{equation}
  \xymatrixcolsep{4pc}
    \myxymatrix{
    G\times G\times X\ar@{->}[r]^(.55){\sfm\times \id_X} \ar@{->}[d]_{\id_G\times \alpha} & G\times X \ar@{->}[d]^{\alpha}\\
    G\times X \ar@{->}[r]^{\alpha} & X
    }
  \hspace{0.5cm}  
  \xymatrixcolsep{4pc}
    \myxymatrix{
    {*}\times X\ar@{->}[r]^{\unit\times \id_X} \ar@{->}[d]_{p_2} & G\times X \ar@{->}[d]^{\alpha}\\
    X\ar@{->}[r]^{\id_X} & X
    }
\end{equation}
\end{definition}

We now need to lift the above definitions to 2-categories. The definition of a 2-group object was introduced in~\cite{Baez:0307200}. Here, we follow closely the presentation in~\cite{Schommer-Pries:0911.2483}.
\begin{definition}
Given a weak 2-category $\CCC$ with finite products, a 2-group object in $\CCC$ is given by an object $C_0$ in $\CCC$ together with 1-morphisms $\sfm:C_0\times C_0\rightarrow C_0$ and $\sfe:*\rightarrow C_0$ as well as invertible 2-morphisms
\begin{equation}
\begin{aligned}
 &\hspace{1.8cm}\sfa:\sfm\circ(\sfm\times 1)\rightarrow \sfm\circ(1\times \sfm)~,\\
 &\sfl:\sfm\circ(\sfe\times 1)\rightarrow \cong\eand \sfr:\sfm\circ(1\times \sfe)\rightarrow \cong~,
\end{aligned}
\end{equation}
where $\cong$ denotes the isomorphisms $*\times C_0\cong C_0\cong C_0\times *$. The morphism $(\pr_1,\sfm):C_0\times C_0\rightarrow C_0\times C_0$ has to be an equivalence and $\sfa$, $\sfl$ and $\sfr$ have to satisfy the (internal) pentagon and triangle identities, cf.~\cite{Schommer-Pries:0911.2483}. 
\end{definition}
\noindent The pentagon and triangle identities are obtain by considering the two obvious morphisms \vspace*{0.2cm}
\begin{equation}
 \myxymatrix{ (\sfm\circ(\sfm\times 1))\circ(\sfm\times 1\times 1)
  \ar@/_4ex/[r]
  \ar@/^4ex/[r]
& \sfm\circ((1\times \sfm)\circ(1\times1\times \sfm))
}\vspace*{0.3cm} 
\end{equation}
and for the triangle identity, we look at\vspace*{0.2cm}
\begin{equation}
\myxymatrix{\sfm\circ((\sfm\circ(1\times \sfe))\times 1)
  \ar@/_4ex/[r]="g1"
  \ar@/^4ex/[r]="g3"
& \sfm\circ(1\times 1)
}\vspace*{0.3cm}
\end{equation}
This yields essentially the diagrams~\eqref{diag:int_cat_pentagon_triangle} with $B_c$ replaced by $\sfm$.

\begin{definition}
 A homomorphism $\phi$ between 2-groups $G$ and $G'$ in a weak 2-category $\CCC$ consists of a 1-morphism $\phi_1:G_1\rightarrow G_2$ and a 2-isomorphism $\phi_2:\sfm'\circ(\phi_1\times \phi_1)\rightarrow \phi_1\circ \sfm$ and $\phi_0:\sfid'\rightarrow \phi_1\circ \sfid$. These have to satisfy three coherence axioms, cf.~\cite{Schommer-Pries:0911.2483} and diagrams~\eqref{diag:int_functor_coherence}.
\end{definition}

\begin{definition}
 Given a 2-group object $\CG$ in a weak 2-category $\CCC$ with finite products, a $\CG$-object in $\CCC$ is an object $\CX$ in $\CCC$ together with a 1-morphism $\alpha:\CG\times \CX\rightarrow \CX$ as well as invertible 2-morphisms $\sfa_\alpha:\alpha\circ(\sfm\times \id_\CX)\rightarrow \alpha\circ(\id_\CG\times \alpha)$ and $\sfl_\alpha:\alpha\circ(\sfe\times \id_\CX)\rightarrow \id_\CX$, such that the following diagrams 2-commute:
\begin{equation}
  \xymatrixcolsep{4pc}
    \myxymatrix{
    \CG\times \CG\times \CX\ar@{->}[r]^(.55){\sfm\times \id_\CX} \ar@{->}[d]_{\id_\CG\times \alpha} & \CG\times \CX \ar@{->}[d]^{\alpha} \ar@{=>}[dl]_{\sfa_\alpha}\\
    \CG\times \CX \ar@{->}[r]_{\alpha} & \CX
    }
  \hspace{0.5cm}  
  \xymatrixcolsep{4pc}
    \myxymatrix{
    {*}\times \CX\ar@{->}[r]^{\sfe\times \id_\CX} \ar@{->}[d]_{p_2} & \CG\times \CX \ar@{->}[d]^{\alpha} \ar@{=>}[dl]_{\sfl_\alpha}\\
    \CX\ar@{->}[r]_{\id_\CX} & \CX
    }
\end{equation}
Moreover, certain coherence axioms for $\sfa_\alpha$ and $\sfl_\alpha$ are satisfied, cf.~\cite{Schommer-Pries:0911.2483} and diagrams~\eqref{diag:int_natural_trafo_coherence}. 
\end{definition}

\subsection{Further useful definitions and proposition}\label{app:B}

To compose bibundles, we needed the notion of coequalizer. 
\begin{definition}
 Given two morphisms $f,g:X_1\rightrightarrows X_0$ in a category $\CCC$, a \uline{coequalizer} is an object $Y\in \CCC$ together with a morphism $y:X_0\rightarrow Y$ such that $y\circ f=y\circ g$. Moreover, we demand that the pair $(y,Y)$ is universal in the sense that for any other such pair $(y',Y')$, there is a unique morphism $u:Y\rightarrow Y'$.
\end{definition}

In the definition of a $\CG$-stack over $X$, we needed the concept of a slice 2-category. First, the following definition.
\begin{definition}
 Given a category $\CCC$ together with an object $a\in\CCC$, the \uline{slice category} $\CCC/a$ has as its objects the class $\CCC(-,a)$. Morphisms between two objects $f:b\rightarrow a$ and $g:c\rightarrow a$ of the slice category are commutative triangles, i.e.\ elements $h$ of $\CCC(b,c)$ with $g\circ h=f$.
\end{definition}
This is generalized to 2-categories as follows.
\begin{definition}
 Given a (weak) 2-category $\CCC$ and an object $a\in \CCC$, the \uline{slice 2-category} $\CCC/a$ consists of the following data. The objects $\CCC/a$ are the 1-morphisms of $\CCC$ with codomain $a$. The 1-morphisms of $\CCC/a$ between objects $f:b\rightarrow a$ and $g:c\rightarrow a$ are pairs $(h,\chi)$, where $h:b\rightarrow c$ is a 1-morphism of $\CCC$ and $\chi:f\cong gh$ is a 2-isomorphism in $\CCC$:
 \begin{equation}
  \myxymatrix{
    b \ar@{->}[rr]^{h} \ar@{->}[dr]_{f} & & c \ar@{->}[dl]^{g}\\
    \ar@{}[urr]^(.4){}="a"^(.8){}="b" \ar@{=>} "a";"b"^{\chi} & a
    }
 \end{equation}
 Finally, consider two 1-morphisms $(h_1,\chi_1)$ and $(h_2,\chi_2)$ between $f:b\rightarrow a$ and $g:c\rightarrow a$ in $\CCC/a$. The 2-morphisms from $(h_1,\chi_1)$ to $(h_2,\chi_2)$ are 2-morphisms $\xi:h_1\Rightarrow h_2$ of $\CCC$ such that $\chi_2(f)=(g\xi)\chi_1(f)$.
\end{definition}

\bibliography{bigone}

\bibliographystyle{latexeu}

\end{document}